\documentclass[
submission
]{dmtcs-episciences}

\pdfoutput=1 % force pdflatex on arXiv
\usepackage[utf8]{inputenc}
\usepackage[T5,T1]{fontenc}
\usepackage{hyperref}
\usepackage{amsmath}
\usepackage{amsthm}
\usepackage{amssymb}
\usepackage{multirow}
\usepackage{subfigure}
\usepackage{array}
\usepackage{cleveref}
\usepackage{csquotes}
\usepackage{mdframed}
\usepackage{mathabx}
\usepackage{stmaryrd}
\usepackage[shortlabels]{enumitem}

\usepackage{tikz}
\usetikzlibrary{shapes,arrows}
\usepackage{tikz-cd}
\tikzstyle{vertex}=[circle,fill=black,minimum size=3pt,inner sep=0pt]
\tikzstyle{bigvertex}=[circle,draw,thick,fill=black!5,minimum size=16pt,inner
sep=0pt]
\newcommand{\partto}{\rightharpoonup}
\newcommand{\partinj}{\rightleftharpoons}

\newcommand{\powerset}{\mathcal{P}}
\newcommand{\naturalN}{\mathbb{N}}

\newcommand{\GrH}{\mathrel{\mathcal{H}}}
\newcommand{\GrL}{\mathrel{\mathcal{L}}}
\newcommand{\GrR}{\mathrel{\mathcal{R}}}

\newcommand{\ttreverse}{\mathtt{reverse}}

\newtheorem{theorem}{Theorem}[section]
\newtheorem{lemma}[theorem]{Lemma}

\newtheorem{corollary}[theorem]{Corollary}

\theoremstyle{definition}
\newtheorem{definition}[theorem]{Definition}
\newtheorem{example}[theorem]{Example}
\newtheorem{remark}[theorem]{Remark}

\newcommand\bumpers{{\triangleright\triangleleft}}
\newcommand\vareps\varepsilon
\newcommand\pol{\rho}
\newcommand\cA{\mathcal{A}}

\newcommand\cG{\mathcal{G}}
\newcommand\cT{\mathcal{T}}

\newcommand\transgraph{\cG}
\newcommand\compgraph{\cG}
\newcommand\tlmon{\mathbf{TL}}
\newcommand\behavmon{\mathbf{Beh}}

\newcommand\setreg{\mathrel{:=}}

\newcommand\deltalr{\delta_\rightarrow}
\newcommand\deltarl{\delta_\leftarrow}
\newcommand\deltall{\delta_\lefttorightarrow}
\newcommand\deltarr{\delta_\righttoleftarrow}

\usepackage[round]{natbib}

\author{{\fontencoding{T5}\selectfont Lê Thành Dũng (Tito) Nguyễn}\affiliationmark{1}\thanks{Supported by the LABEX MILYON (ANR-10-LABX-0070) of Université de Lyon, within the program \enquote{Investissements d'Avenir} (ANR-11-IDEX-0007) operated by the French National Research Agency (ANR).}
\and
{Camille Noûs}\affiliationmark{2}
  \and Cécilia Pradic\affiliationmark{3}}
\title{Two-way automata and transducers with planar behaviours are aperiodic}
\affiliation{
  Laboratoire de l'informatique du parallélisme (LIP), École normale supérieure de Lyon, France\\
  Laboratoire Cogitamus\\
  Department of Computer Science, Swansea University, Wales}
\begin{document}
% This is only used if you are compiling for a volume before vol 25
\publicationdetails{VOL}{2023}{ISS}{NUM}{SUBM}
% This is the new form of collecting the data, starting with vol 25
% Cécilia: j'arrive pas à faire compiler avec ce truc...
%\publicationdata
%{vol. 25:3 special issue for main purpose}
%{2022}
%{1}
%{10.46298/dmtcs.10472}
%%{1998-10-14; 1998-10-14; 2002-07-19; 2014-02-05; 2015-09-09; 2022-12-25}
%%{2022-12-3}
%{2022-12-3; None}
%{2023-1-1}

\maketitle

\begin{abstract}
  We consider a notion of planarity for two-way finite automata and transducers,
  inspired by Temperley-Lieb monoids of planar diagrams.
  We show that this restriction captures star-free languages and
  first-order transductions.
\end{abstract}

\section{Introduction}

This paper considers a notion of \emph{planarity} -- in the topological sense --
for two-way deterministic finite automata. To the best of our knowledge, little
has been said about it; \citet{HinesPlanar} suggested the definition in a
workshop talk, with motivations coming from outside automata theory, but he did
not investigate it further. Unlike the notion of \enquote{planar automaton} that
seems to occur more frequently in the literature,
e.g.~\citep{inhplanaraut,BonfanteD19}, ours is concerned not with
the state transition chart, but with another graphical representation of the
behaviour of a two-way automaton.

\begin{example}\label{ex:intro}
  The following diagram corresponds to a two-way automaton -- defined in
  \Cref{ex:2dfaplanar} -- running on an input containing a factor $abca$. Assuming
  that the automaton arrives on this factor from the left in state
  $0^\rightarrow$, its run will continue as indicated by the path in red,
  exiting on the right in state $3^\rightarrow$.
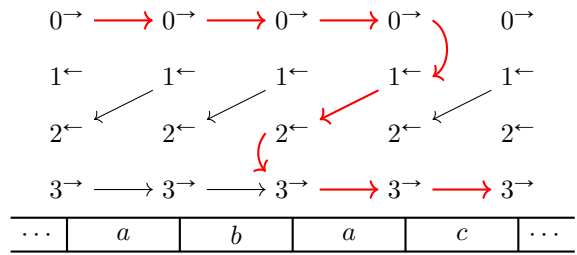
  \begin{center}
    \begin{tikzpicture}[scale=0.75]
      \draw[thick] (-1,1) -- (9,1);
      \draw[thick] (9,0.4) -- (-1,0.4);
      \draw[thick] (0,0.4) -- (0,1);
      \draw[thick] (2,0.4) -- (2,1);
      \draw[thick] (4,0.4) -- (4,1);
      \draw[thick] (6,0.4) -- (6,1);
      \draw[thick] (8,0.4) -- (8,1);
      \node at (-0.5,0.7) {$\cdots$};
      \node at (1,0.7) {$a$};
      \node at (3,0.7) {$b$};
      \node at (5,0.7) {$a$};
      \node at (7,0.7) {$c$};
      \node at (8.5,0.7) {$\cdots$};
      \node (q01) at (0,4.5) {$0^\rightarrow$};
      \node (q11) at (0,3.5) {$1^\leftarrow$};
      \node (q21) at (0,2.5) {$2^\leftarrow$};
      \node (q31) at (0,1.5) {$3^\rightarrow$};
      \node (q02) at (2,4.5) {$0^\rightarrow$};
      \node (q12) at (2,3.5) {$1^\leftarrow$};
      \node (q22) at (2,2.5) {$2^\leftarrow$};
      \node (q32) at (2,1.5) {$3^\rightarrow$};
%      {\draw[->] (q11) -- (q12);}
      {\draw[->, thick, color=red] (q01) -- (q02);}
      \draw[->] (q12) -- (q21);

      \node (q03) at (4,4.5) {$0^\rightarrow$};
      \node (q13) at (4,3.5) {$1^\leftarrow$};
      \node (q23) at (4,2.5) {$2^\leftarrow$};
      \node (q33) at (4,1.5) {$3^\rightarrow$};
      \draw[->] (q13) -- (q22);
%      {\draw[->] (q12) -- (q13); \draw[->] (q33) -- (3.3,1.5);}
      {\draw[->, thick, color=red] (q02) -- (q03);
        \draw[->, thick, color=red] (q23.west) to [bend right=40] (q33.north west);}

      \node (q04) at (6,4.5) {$0^\rightarrow$};
      \node (q14) at (6,3.5) {$1^\leftarrow$};
      \node (q24) at (6,2.5) {$2^\leftarrow$};
      \node (q34) at (6,1.5) {$3^\rightarrow$};
%      {\draw[->] (q13) -- (q14); \draw[->] (q24) -- (q33);}
      {\draw[->, thick, color=red] (q03) -- (q04);
        \draw[->, thick, color=red] (q14) -- (q23);}

      \node (q05) at (8,4.5) {$0^\rightarrow$};
      \node (q15) at (8,3.5) {$1^\leftarrow$};
      \node (q25) at (8,2.5) {$2^\leftarrow$};
      \node (q35) at (8,1.5) {$3^\rightarrow$};
 %     {\draw[->] (q14.east) to [bend left=40] (q24.east);}
      {\draw[->, thick, red] (q04.east) to [bend left=60] (q14.east);}
      \draw[->] (q15) -- (q24);

      \draw[->] (q31) -- (q32);
      \draw[->] (q32) -- (q33);
      \draw[->, thick, red] (q33) -- (q34);
      \draw[->, thick, red] (q34) -- (q35);
    \end{tikzpicture}
  \end{center}
\end{example}
By \enquote{planarity}, we refer to the lack of
crossings\footnote{There is an unfortunate clash of terminology here with the \emph{crossing sequences} from~\cite{shepherdson59}. Therefore, we shall avoid the word \enquote{crossing} altogether.} in such diagrams drawn on the plane. This is the case in the above example.

Our main result entails
that planar deterministic two-way automata recognize exactly the class of
\emph{star-free} languages.
We actually show a strengthening of this result on \emph{transducers} -- finite-state machines that compute string-to-string functions.
  \begin{theorem}
  \label{thm:main}
    Let $f\colon \Sigma^* \partto \Gamma^*$. The following are equivalent:
    \begin{itemize}
    \item $f$ is a \emph{first-order transduction}.
    \item $f$ is computed by some \emph{planar deterministic} two-way finite transducer.
    \item $f$ is computed by some \emph{planar reversible}\footnote{A reversible machine is both deterministic and \enquote{co-deterministic}, as we will explain.} two-way finite transducer.
    \end{itemize}
  \end{theorem}

  By forgetting the output information of a two-way transducer, one may obtain a two-way automaton recognizing the
  domain of the partial function computed by the transducer. Thus, we immediately get \Cref{cor:main} below for automata
  and languages. To make sense of this, keep in mind that first-order transductions are to regular functions -- those computed by deterministic two-way transducers (2DFTs), without the planarity condition -- what star-free languages are to regular languages \cite[see e.g.][p.~9]{MuschollPuppis}. (While the name \enquote{first-order} refers to logic, we will not use logic in this paper.)
  \begin{corollary}
  \label{cor:main}
    Let $L \subseteq \Sigma^*$. The following are equivalent:
    \begin{itemize}
    \item $L$ is a star-free language.
    \item $L$ is recognized by some \emph{planar deterministic} two-way finite automaton (planar 2DFA).
    \item $L$ is recognized by some \emph{planar reversible} two-way finite automaton.
    \end{itemize}
  \end{corollary}
  
  \subparagraph*{Plan of the paper} After a discussion of motivations
  and related work, we introduce in \Cref{sec:prelim}
  the automata model and the notion of planarity under scrutiny. We
  then embark on the proof of \Cref{thm:main}, which contains two main
  thrusts corresponding to the equivalence between the first and last
  items.
  
  \Cref{sec:aperiodic} shows that functions
  computed by planar 2DFTs are first-order transductions. To do so,
  we appeal to a result of~\citet{cartondartoisaperiodic}
  stating that \emph{aperiodic} 2DFTs   are exactly as expressive as first-order
  transductions (extending Schützenberger's seminal theorem relating star-free languages and aperiodic monoids, whose history is recounted in \citep{StraubingSIGLOG}). In order to show that the monoid of behaviours of a planar 2DFT is necessarily
  aperiodic, we simply embed it into a generic monoid that depends
  only on the (ordered and directed) state space of the 2DFT rather than the fine
  details of its transition functions. The key feature that makes this so simple is that
  planarity is a property that is preserved by composition
  of transitions (whereas aperiodic elements of a monoid are, in general, not closed under product).
  
  Then, in \Cref{sec:krohnrhodes}, we show that the functions realizable by
  planar reversible two-way transducers (2RFTs) form a class that is closed under composition and contains some \enquote{simple} functions.
  We can then conclude by factoring any first-order transduction into a composition of such simple functions (this involves the decomposition theorem of \citet{KrohnRhodes}). Here the interesting
  aspect is that planarity of transitions is preserved by known methods~\citep{ReversibleTransducers}
  to compose 2RFTs.

\subsection{Related work}

\paragraph{Planar diagram monoids and geometry of interaction}
A device amounting to an \enquote{undirected} variant of our planar 2RFA was first proposed by \citet{HinesPlanar}.
He had previously connected~\citep{Hines} the monoids of 2DFA and 2RFA behaviours to a category-theoretic account of the so-called \emph{geometry of interaction} (GoI)\footnote{Further connections between categorical versions of the GoI and automata theory have been investigated by~\citet{Katsumata08} and by~\citet{MAHORS}. We may also note that \emph{inverse monoids} have been linked to both GoI~\citep{GoubaultLarrecqGoI} and two-way automata~\citep*{PreRational}.} semantics of linear logic. Then planarity comes up naturally when considering GoI for non-commutative linear logic, as noted for instance by~\citet{AbramskyTemperleyLieb} who also explains how the monoids of endomorphisms of his categorical semantics coincide with \citeauthor{Kauffman}'s \citeyearpar{Kauffman} diagram monoids. (The monoids of behaviours in \citep{HinesPlanar} are finite quotients\footnote{An element of a Kauffman monoid consists not only of a planar diagram, but also of a natural number of free-floating cycles, making the monoid countably infinite; the quotient throws these cycles away.} of Kauffman monoids, called \enquote{Jones monoids} or \enquote{Temperley-Lieb\footnote{The \emph{Temperley-Lieb algebras} mentioned in the titles of \citep{HinesPlanar,AbramskyTemperleyLieb} admit a presentation based on Kauffman monoids. They appear in knot theory and statistical physics.} monoids} \citep[see e.g.][]{PresentationsTL}.)

In an earlier work~\citep{aperiodic}, we exhibited a connection between non-commutative linear logic and star-free languages. It is thanks to this that we could guess what effect the planarity restriction would have on the computational power of two-way automata and transducers. Relatedly, \citet{Auinger14} has shown that Jones monoids generate precisely the pseudovariety of aperiodic monoids -- a result that is morally very close to our characterization of star-free languages by planar 2RFA, and whose proof uses the algebraic version of the Krohn--Rhodes decomposition.

\paragraph{Relationship with previous automata models}

As we already said in our opening words, our planar behaviours are entirely unrelated to the \emph{planar automata} of~\citet{inhplanaraut}. Their
notion of planarity relates to usual one-way machines when regarded as graphs, so
it includes, for instance, the minimal DFA of the non-star-free language $(aa)^*$.

That said, a special case of our planarity condition has appeared before in the literature on automata. Indeed, when instantiated to deterministic one-way automata, it becomes \emph{monotonicity}: all transitions are partial monotone functions with respect to a common linear order on states. The variant with total instead of partial transitions has been introduced and dubbed \enquote{monotonic automata} by \citet{SynchroMonotonic}\footnote{See for instance~\citep*[Section~7]{SynchroGame} for an overview of subsequent work on monotonic automata.} in the context of synchronization problems. The study of the monoid $\mathcal{O}_n$ of all partial monotone functions on the finite linear order $\{1 < 2 < \dots < n\}$ -- which contains the monoids of behaviours of (partial) monotonic automata, i.e.\ one-way planar automata, with $n$ states -- is actually even older. In particular, \citet{Higgins95} exhibits an aperiodic monoid that does not divide any $\mathcal{O}_n$ for $n\in\naturalN$, which entails that some star-free language cannot be recognized by any partial monotonic automaton. Thus, two-wayness is necessary to recognize all star-free languages with planar automata.

As for our definition of reversibility for automata, it does coincide with the widespread one in the literature. \citet*{ReversibleTransducers} were the first to study it in the context of two-way transducers, but there had been several earlier works proving that reversibility does not affect the power of two-way automata, and of their generalizations to trees and graphs -- for an overview, see~\citep[Introduction]{ReversibleGraphWalking}.
Finally, on a technical level, this paper is very much tied to the existing theory
of machine models for first-order transductions as developed in~\citep*{FOSST,cartondartoisaperiodic,AperiodicSST,ListFunctions}.

\section{Preliminaries}
\label{sec:prelim}

\subparagraph*{Notations} Alphabets $\Sigma,\Gamma$ consist of non-empty finite sets.
By convention, we assume that we have distinguished elements $\triangleright$ and $\triangleleft$
that do not belong to any alphabet and we set $\Sigma_\bumpers = \Sigma \cup \{\triangleright,\triangleleft\}$.
Given two sets $A$ and $B$, we write $A \partto B$ for the set of partial functions from $A$ to $B$
and $A \partinj B$ for partial injections from $A$ to $B$.
Finally, given a relation $\to$, we write $\to^*$ for its reflexive transitive closure.

\begin{definition}
\label{def:2nfa}
A two-way nondeterministic automaton (2NFA) $\cA$ over the alphabet $\Sigma$ consists of:
\begin{itemize}
\item A finite set $Q$ of \emph{states} equipped with a \emph{direction map} $\pol \colon Q \to \{-1, 1\}$.\\
The pair $(Q,\pol)$ is sometimes called
a \emph{directed set of states} and we write $Q^\rightarrow$ for $\pol^{-1}(1)$ and $Q^\leftarrow$ for $\pol^{-1}(-1)$.
Elements of $Q^\rightarrow$ are called \emph{forward} states and those of $Q^\leftarrow$ \emph{backward} states.
\item An \emph{initial state} $q_0 \in Q^\rightarrow$ and a set\footnote{We could take $F$ to be a singleton or replace $q_0$ by a
    set without changing any result of the paper.} of \emph{final states} $F \subseteq Q$.
\item A family of \emph{transition relations} $\delta\colon\Sigma_{\bumpers} \to \powerset(Q\times Q)$.
\end{itemize}
A configuration of the automaton is a triple $(u,q,v) \in \Sigma_\bumpers^* \times Q \times \Sigma_\bumpers^*$.
The immediate successor relation $\to_\delta$ is defined as follows
\[
\begin{array}{lcl !\qquad l}
(u,q,av) &\to_\delta& (ua,r,v) &\text{whenever $(q,r) \in \delta(a)$, $q \in Q^\rightarrow$ and $r \in Q^\rightarrow$} \\
(u,q,av) &\to_\delta& (u,r,av) &\text{whenever $(q,r) \in \delta(a)$, $q \in Q^\rightarrow$ and $r \in Q^\leftarrow$} \\
(ua,q,v) &\to_\delta& (u,r,av) &\text{whenever $(q,r) \in \delta(a)$, $q \in Q^\leftarrow$ and $r \in Q^\leftarrow$} \\
(ua,q,v) &\to_\delta& (ua,r,v) &\text{whenever $(q,r) \in \delta(a)$, $q \in Q^\leftarrow$ and $r \in Q^\rightarrow$} \\
\end{array}
\]
A word $w \in \Sigma^*$ is said to be accepted by $\cA$ if we have
$(\varepsilon,q_0,\triangleright w \triangleleft) \to_\delta^* (\triangleright w \triangleleft,q_f,\varepsilon)$ for some $q_f \in F$. That is, a run of the 2NFA starts in the initial state on the left of the input string, and to accept a word is to reach a final state after overrunning the right-hand side
marker.
\end{definition}

We consider only \emph{deterministic} machines in this paper, i.e.\ we require that the
transition relations are actually partial functions. In fact, we shall study
a further restriction that is \emph{reversibility}.
At the formal level this
amounts to requiring that the inverse of the transition functions
also be partial functions.

\begin{definition}
\label{def:rev}
An automaton $(Q,\pol,q_0,F,\delta)$ over $\Sigma$ is called
\emph{deterministic} (resp.\ \emph{reversible}) if $\delta$ maps each letter $a\in\sigma$
to a \emph{partial function} (resp. \emph{injection}).
\end{definition}

The languages recognized by 2DFA are exactly the regular languages -- that is, \emph{2DFA are equivalent to one-way automata}~\citep[Theorem~15]{RabinScott}. Nowadays, the best-known construction for translating 2DFA to one-way DFA is due to~\citet{shepherdson59}; it is the one that
we shall work with\footnote{Technically speaking, we use two-sided
  behaviours as in~\citep{Birget} so as to discuss monoids while Shepherdson only
  considered one-sided behaviours; the combinatorics remain similar.}.

\begin{example}
\label{ex:2dfaplanar}
We define here the two-way automaton whose behaviour on a sample input is depicted in \Cref{ex:intro}.
The input alphabet is $\{a,b,c\}$, and the set of states is $Q^\rightarrow = \{0^\rightarrow,3^\rightarrow\}$,
$Q^\leftarrow = \{1^\leftarrow,2^\leftarrow\}$ and $F = \{3^\rightarrow\}$.
The transition function is specified the following input/output table:
\[
\begin{array}{c|ccccc}
 & a & b & c & \triangleright & \triangleleft \\
\hline
0^\rightarrow &
  0^\rightarrow & 0^\rightarrow & 1^\leftarrow & 0^\rightarrow & \\
1^\leftarrow &
  2^\leftarrow & 2^\leftarrow & 2^\leftarrow & & \\
2^\leftarrow &
   & 3^\rightarrow  \\
3^\rightarrow &
  3^\rightarrow & 3^\rightarrow & 3^\rightarrow & 3^\rightarrow & 3^\rightarrow \\
\end{array}
\]
It recognizes the language $\{a,b\}^*b\{a,b\}c\{a,b,c\}^*$ by
first looking for the first occurrence of $c$, then doubling
back to check a $b$ occurs two positions earlier. This automaton:
\begin{itemize}
  \item is deterministic, as seen from the fact that each cell of the table contains at most one output;
  \item is not reversible, because the transition specified for $b$ has two different states sent to $3^\rightarrow$.
\end{itemize}
\end{example}

We now formalize the central restriction on 2DFA at play in this paper, which
is planarity. As per \Cref{fig:planarity}, it means that we want to enforce that each
transition profile, as drawn in our example pictures, does not contain any sort
of crossings. \Cref{ex:2dfaplanar} is a (paradigmatic) example of a planar 2DFA.

\begin{definition}
\label{def:2dfa-transitionprofile}
Let $(Q,\pol)$ be a directed set of states.
For any binary relation $f \subseteq Q \times Q$, we define 
the \emph{transition profile} $\transgraph(\pol,f)$ as
the directed graph with vertices $Q \times \{-1,1\}$ and edges
\[ \{(q,-\pol(q)) \to (r,\pol(r)) \mid (q,r) \in f \} \]
\end{definition}

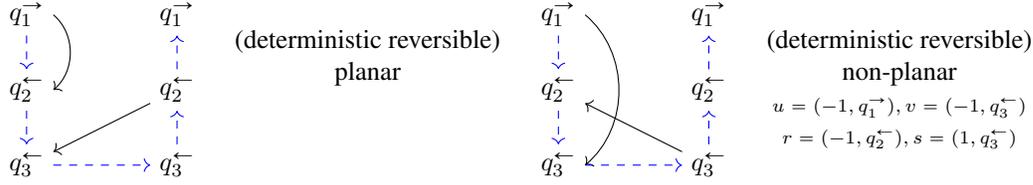
\begin{figure}
  \begin{tabular}{cccc}
      \multirow{3}*{\begin{tikzpicture}
      \node (q11) at (0,3) {$q_1^\rightarrow$};
      \node (q21) at (0,2) {$q_2^\leftarrow$};
      \node (q31) at (0,1) {$q_3^\leftarrow$};
      \node (q12) at (2,3) {$q_1^\rightarrow$};
      \node (q22) at (2,2) {$q_2^\leftarrow$};
      \node (q32) at (2,1) {$q_3^\leftarrow$};
      \draw[->] (q11.east) to [bend left=50] (q21.east);
      \draw[->] (q22) -- (q31);
      {
        \draw[->, dashed, color=blue] (q11) -- (q21);
        \draw[->, dashed, color=blue] (q21) -- (q31);
        \draw[->, dashed, color=blue] (q31) -- (q32);
        \draw[->, dashed, color=blue] (q32) -- (q22);
        \draw[->, dashed, color=blue] (q22) -- (q12);
      }
      \end{tikzpicture}}
      &
      & 
      \multirow{3}*{\begin{tikzpicture}
      \node (q11) at (0,3) {$q_1^\rightarrow$};
      \node (q21) at (0,2) {$q_2^\leftarrow$};
      \node (q31) at (0,1) {$q_3^\leftarrow$};
      \node (q12) at (2,3) {$q_1^\rightarrow$};
      \node (q22) at (2,2) {$q_2^\leftarrow$};
      \node (q32) at (2,1) {$q_3^\leftarrow$};
      \draw[->] (q11.east) to [bend left=50] (q31.east);
      \draw[->] (q32) -- (q21);
      {
        \draw[->, dashed, color=blue] (q11) -- (q21);
        \draw[->, dashed, color=blue] (q21) -- (q31);
        \draw[->, dashed, color=blue] (q31) -- (q32);
        \draw[->, dashed, color=blue] (q32) -- (q22);
        \draw[->, dashed, color=blue] (q22) -- (q12);
      }
      \end{tikzpicture}}
      &\\ &
      (deterministic reversible) &&
      (deterministic reversible) 
      \\
      &
      planar 
      && non-planar\\
      &&&{\scriptsize $u = (-1,q_1^\rightarrow), v = (-1,q_3^\leftarrow)$}\\
      &&&{\scriptsize $r = (-1,q_2^\leftarrow), s = (1,q_3^\leftarrow)$}\\
      \\  \end{tabular}
\caption{Planar and non-planar behaviours for
$Q^\rightarrow = \{q_1^\rightarrow\}$ and $Q^\leftarrow = \{q_2^\leftarrow,q_3^\leftarrow\}$.}
\label{fig:planarity}
\end{figure}

\begin{definition}
\label{def:planarity}
Let $Q$ be a finite set. For every total order $<$ over $Q$,
extend $<$ to a total order over $Q \times \{-1,1\}$
by setting $(q,i) < (r,j)$ if and only if either
\[
\text{$i = -1$ and $j = 1$}
\qquad \text{or} \qquad
\text{$q < r$ and $i = j = 1$}
\qquad \text{or} \qquad
\text{$q > r$ and $i = j = -1$}
\]
Call a transition $f \in \powerset(Q^2)$ \emph{planar} with
respect to the order $<$ and polarity function $\pol$ if
there is no ordered sequence of vertices $u < r < v < s$ such that either $u \to v$ or $v \to u$ is  in $\transgraph(\pol,f)$ and either $r \to s$ or $s \to r$ is in $\transgraph(\pol,f)$.

A 2NFA $(Q,\pol,q_0,F,\delta)$ over alphabet $\Sigma$ is called planar when there is a total order $<$ over $Q$
such that, for every input letter $a \in \Sigma$, $\delta(a)$ is planar with respect to $<$ and $\pol$.
\end{definition}

\begin{remark}\label{rem:combinatorial-map}
  This definitions says that we must choose a total order on states which will constrain how all transitions will be drawn in the plane. This is stronger than merely requiring each transition profile to individually be a planar graph. In fact, it corresponds to the planarity of a \emph{combinatorial map} \citep[see][Chapter~4]{MoharThomassen} whose underlying graph is obtained from the transition profile by forgetting the edge directions.
\end{remark}

Our results will primarily concern two-way \emph{transducers}, a notion that generalizes two-way automata
by realizing binary relations between $\Sigma^*$ and $\Gamma^*$ instead of languages, i.e. sets of words.
Since we are going to restrict our attention to deterministic machines the relations
will always be partial functions $\Sigma^* \partto \Gamma^*$.

\begin{definition}
\label{def:2nft}
A two-way nondeterministic transducer (2NFT) $\cT$ with input alphabet $\Sigma$ and output alphabet
$\Gamma$ consists of a tuple $(Q, \pol, q_0, F, \delta)$ with $(Q,\pol)$ a directed set of states,
$q_0 \in Q^\rightarrow$, $F \subseteq Q$ and, most importantly, with $\delta\colon \Sigma \to \powerset(Q \times \Gamma^* \times Q)$.

The notion of configuration is defined in the same way as \Cref{def:2nfa}, except that now the immediate successor relation $\xrightarrow{w}_\delta$
is tagged by a word $w \in \Gamma^*$. It is defined in the expected way, analogous to the case of 2NFA.
We then say that $(u,v)$ is in the relation realized by $\cT$ if we have a sequence of configurations
$(\vareps,q_0,\triangleright u \triangleleft) \xrightarrow{v_1} \dots \xrightarrow{v_n} (\triangleright u \triangleleft,q_f,\vareps)$ such that $q_f \in F$ and $v = v_1 \ldots v_n$.

%Given a 2NFT $\cT = (Q, \pol, q_0, F, \delta)$ with input alphabet $\Sigma$,
One can consider the
2NFA $\lfloor \cT \rfloor$ over $\Sigma$ obtained by forgetting the output, i.e. $(Q, \pol, q_0, F, \lfloor \delta \rfloor)$
where $\lfloor \delta \rfloor(a) = \{ (q,r) \mid \exists w. ~ (q,w,r) \in \delta\}$
and check that the language recognized by $\lfloor \cT \rfloor$
is exactly the domain of the relation defined by $\cT$.
We can use the map $\cT \mapsto \lfloor \cT \rfloor$ to lift all notions relevant to automata to transducers:
we say that $\cT$ is deterministic/reversible/planar when $\lfloor \cT \rfloor$ is.
\end{definition}

A 2DFT $\cT$ with input alphabet $\Sigma^*$ and output alphabet $\Gamma^*$
defines a partial function $\Sigma^* \partto \Gamma^*$. Let us write
$\cT\colon \Sigma^* \partto \Gamma^*$ and conflate $\cT$ with the
partial function it computes in notations when convenient.

\section{Planar 2DFTs can only compute first-order transductions}
\label{sec:aperiodic}

To show that planar 2DFTs can only compute
first-order transductions, we use a characterization of the latter due to \citet{cartondartoisaperiodic}. To recall it, we need to start with \emph{aperiodic monoids}, the algebraic counterpart of star-free languages (as mentioned in the introduction).
\begin{definition}
\label{def:aperiodicity}
A finite monoid $M$ is aperiodic if there is $n \in \naturalN$
such that for every $x \in M$, $x^{n+1} = x^n$.
\end{definition}

To make full use of this algebraic condition, we associate finite monoids to our machines. The idea is to look at the
transition profiles generated by the transition relations $\delta$ and
their composition.

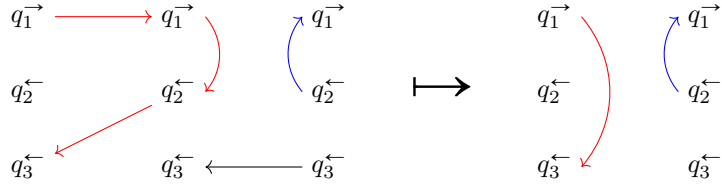
\begin{figure}[h]
  \begin{center}
    \begin{tikzpicture}
      \node (q13) at (4,3) {$q_1^\rightarrow$};
      \node (q23) at (4,2) {$q_2^\leftarrow$};
      \node (q33) at (4,1) {$q_3^\leftarrow$};

      \node (q14) at (6,3) {$q_1^\rightarrow$};
      \node (q24) at (6,2) {$q_2^\leftarrow$};
      \node (q34) at (6,1) {$q_3^\leftarrow$};
      \draw[->, color=red] (q13) -- (q14);
      \draw[->, color=red] (q24) -- (q33);

      \node (q15) at (8,3) {$q_1^\rightarrow$};
      \node (q25) at (8,2) {$q_2^\leftarrow$};
      \node (q35) at (8,1) {$q_3^\leftarrow$};
      \draw[->, color=red] (q14.east) to [bend left=40] (q24.east);
      \draw[->] (q35) -- (q34);
      \draw[->, color=blue] (q25.west) to [bend left=40] (q15.west);

      \node at (9.5,2) {\Huge$\mapsto$};

      \node (q11) at (11,3) {$q_1^\rightarrow$};
      \node (q21) at (11,2) {$q_2^\leftarrow$};
      \node (q31) at (11,1) {$q_3^\leftarrow$};
      \node (q12) at (13,3) {$q_1^\rightarrow$};
      \node (q22) at (13,2) {$q_2^\leftarrow$};
      \node (q32) at (13,1) {$q_3^\leftarrow$};
      \draw[->, color=red] (q11.east) to [bend left=40] (q31.east);
      \draw[->, color=blue] (q22.west) to [bend left=40] (q12.west);
    \end{tikzpicture}
  \end{center}
  \caption{Compostition of two transitions.}
\label{fig:comptrans}
\end{figure}
\begin{definition}
\label{def:behavmonoid}
Consider a directed set of states $(Q,\pol)$.
Given two transitions $f, g \in \powerset(Q^2)$,
consider the digraph $\compgraph(\rho,f,g)$ by
taking the disjoint union of $\transgraph(\rho,f)$
and $\transgraph(\rho,g)$
and identifying $Q \times \{1\}$ in the former with $Q \times \{-1\}$ in the latter.
(Formally, $\compgraph(\rho,f,g)$ is defined as having
vertices $Q \times \{-1,0,1\}$ and the edge $(u,i) \to (v,j)$
if and only if either $(u,2i+1)\to(v,2j+1))$ is
an edge of $\transgraph(\rho,f)$ or $(u,2i-1)\to(v,2j-1))$
is in $\transgraph(\rho,g)$.)

The \emph{composition} $f * g$ is then defined as the
transition relating $q$ to $r$ if and only if there is a
path from $(q,-\rho(q))$ to $(r,\rho(r))$ in the digraph $\compgraph(\rho,f,g)$.
This way, $\transgraph(\rho,\,f*g)$ contains an edge $u\to v$ if and only if $\compgraph(\rho,f,g)$ contains a path from $u$ to $v$ (see \Cref{fig:comptrans} for an illustration).
Composition can be shown to be associative, with the identity relation as unit element \citep[see e.g.][]{Birget}: this yields a finite monoid
$\behavmon(Q,\pol)$ with carrier $\powerset(Q^2)$.

Given a 2NFA $\cA = (Q, \pol, q_0, F, \delta)$ over alphabet $\Sigma$,
define $\behavmon(\cA)$
as the least submonoid of $\behavmon(Q,\pol)$ containing $\{\delta(a) \mid a \in \Sigma\}$.
For 2NFTs, set $\behavmon(\cT) = \behavmon(\lfloor\cT\rfloor)$.
\end{definition}

When $\cT$ is a 2DFT, \cite[Def.~4]{cartondartoisaperiodic} define its transition monoid $\behavmon(\cT)$
in the same way\footnote{The presentation is different, but
the result is easily seen to be isomorphic.}, so we may use their characterization as our reference definition\footnote{The main result of \citep{cartondartoisaperiodic} is that this definition is equivalent to another one based on first-order logic, hence the name of the function class. The papers cited at the very end of the introduction provide several alternative characterizations.} of first-order transductions.

\begin{definition}
  A first-order transduction is a string-to-string function computed by a 2DFT $\cT$ whose monoid of behaviours $\behavmon(\cT)$ is aperiodic.
\end{definition}

Next, we want to show that the monoid of behaviours of a planar 2DFT
is necessarily aperiodic. As submonoids of aperiodic monoids are also
aperiodic, it suffices to show that the following class of monoids defined from
directed set of states are aperiodic.

\begin{definition}
\label{def:tlmon}
Given a directed set of states $(Q,\pol)$ and a total order $<$ on $Q$, call $\tlmon(Q,\pol,<)$ the subset of $\behavmon(Q,\pol)$
containing only transitions that are both deterministic (i.e.\ partial functions) and planar with respect to $\pol$ and $<$.
\end{definition}

\begin{theorem}
\label{thm:tlmonisaperiodic}
 $\tlmon(Q,\pol,<)$ is an aperiodic submonoid of $\behavmon(Q,\pol)$.
\end{theorem}
\begin{proof}
  First, we show that it is a submonoid. We leave the reader to mechanically check that the unit element is deterministic planar. Concerning closure under composition, it is topologically intuitive:\footnote{Our argument here is somewhat informal. It could be made rigorous using the Heffter--Edmonds--Ringel principle \citep[cf.][Theorem~3.2.4]{MoharThomassen} relating combinatorial maps and actual topology. In principle, there should probably also exist a purely combinatorial proof from \Cref{def:planarity}, but it would likely be cumbersome.} for $f,g\in\tlmon(Q,\pol,<)$, the combinatorial map (cf.~\Cref{rem:combinatorial-map}) $\compgraph(\rho,f,g)$ is planar because it is obtained by gluing two planar maps along a common boundary, identifying vertices; by determinism, it consists of a disjoint union of rooted trees that can be drawn in the plane in a pairwise non-crossing fashion, and this drawing is homotopy equivalent to $\transgraph(\rho,\,f*g)$ (fixing the vertices $Q\times\{-1,1\}$), by contracting the trees to stars. Note that determinism is necessary here: \Cref{fig:nondetcex} shows that planar non-deterministic behaviours do not compose.
  \begin{figure}[h]
    \begin{center}
      \begin{tikzpicture}
        \node (q13) at (4,3) {$q_1^\rightarrow$};
        \node (q23) at (4,2) {$q_2^\rightarrow$};
          
        \node (q14) at (6,3) {$q_1^\rightarrow$};
        \node (q24) at (6,2) {$q_2^\rightarrow$};
        \draw[->, color=blue] (q13) -- (q14);
        \draw[->, color=blue] (q23) -- (q14);
  
        \node (q15) at (8,3) {$q_1^\rightarrow$};
        \node (q25) at (8,2) {$q_2^\rightarrow$};
        \draw[->, color=blue] (q14) -- (q15);
        \draw[->, color=blue] (q14) -- (q25);
  
        \node at (9.5,2.5) {\Huge$\mapsto$};
  
        \node (q11) at (11,3) {$q_1^\rightarrow$};
        \node (q21) at (11,2) {$q_2^\leftarrow$};
        \node (q12) at (13,3) {$q_1^\rightarrow$};
        \node (q22) at (13,2) {$q_2^\leftarrow$};
        \draw[->, color=blue] (q11) -- (q12);
        \draw[->, color=blue] (q11) -- (q22);
        \draw[->, color=blue] (q21) -- (q12);
        \draw[->, color=blue] (q21) -- (q22);
      \end{tikzpicture}
    \end{center}
    \caption{A non-planar composition of two planar non-deterministic behaviours.}
  \label{fig:nondetcex}
  \end{figure}
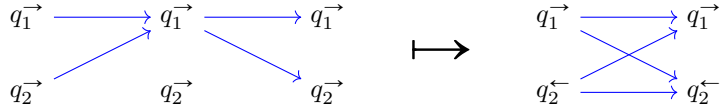
  
  Let us now show aperiodicity. Let $f \in \tlmon(Q,\pol,<)$ and let us partition $f$ into
  \[ f_\righttoleftarrow = f \cap (Q^\rightarrow \times Q^\leftarrow) \qquad f_\lefttorightarrow = f \cap (Q^\leftarrow \times Q^\rightarrow) \qquad f_\leftrightarrow = f \setminus (f_\righttoleftarrow \cup f_\lefttorightarrow)\]
  We shall consider variants of $f_\righttoleftarrow$ and $f_\lefttorightarrow$ where \enquote{partial identities} have been added:
  \[
    f_\dashv = f_\righttoleftarrow \cup \{ (q,q) \in Q^2 \mid \not\exists r : (q,r) \in f_\righttoleftarrow \}\qquad
    f_\vdash = f_\lefttorightarrow \cup \{ (q,q) \in Q^2 \mid \not\exists r : (q,r) \in f_\lefttorightarrow \}
  \]
  The point is that we have the following identity (which, for now, does not depend on planarity):
  \[ f = f_\dashv * f_\leftrightarrow * f_\vdash \quad\text{hence}\quad \forall n\in\naturalN,\; f^{n+1} = f_\dashv * (f_\leftrightarrow * f_\vdash * f_\dashv)^n * f_\leftrightarrow * f_\vdash \]
  Let $g = f_\leftrightarrow * f_\vdash * f_\dashv$. The above identity shows that if the sequence $(g^n)_{n\geq1}$ is eventually constant, then so is $(f^n)_{n\geq1}$ -- and the latter is what we want, according to the definition of aperiodicity.

  Observe that $g \subseteq Q^\rightarrow \times Q^\rightarrow \cup Q^\leftarrow \times Q^\leftarrow$. Thus, $g$ is the union of a partial function $g_\rightarrow\colon Q^\rightarrow \partto Q^\rightarrow$ with the transpose of a partial function $g_\leftarrow\colon Q^\leftarrow \partto Q^\leftarrow$, and these two components do not interact in a composition. Since $g$ is planar (as a product of planar transitions), these partial functions are both monotone. So, to conclude, it suffices to invoke the well-known fact that the monoid of partial monotone functions on a finite total order is aperiodic \citep[see e.g.][Lemma~3.2]{aperiodic}.
\end{proof}

\begin{remark}
  An alternative proof of aperiodicity can also be obtained by characterizing some of
\emph{Green's relations} over $\tlmon(Q, \pol, <)$, which we sketch now.
The $\GrR$-class of any element of $f \in \tlmon(Q, \pol, <)$ (which corresponds
to the left ideal generated by $f$) is fully determined
by $f_\righttoleftarrow$ and the set $f_r \in \powerset(\powerset(Q))$ defined as
\[ 
  f_r \quad = \quad \{f_\leftrightarrow(\{q\}) \mid q \in Q^\leftarrow\} \cup \{f_\leftrightarrow^{-1}(\{q\}) \mid q \in Q^\rightarrow\}
\]
The intuition is that $f_\righttoleftarrow$ and $f_r$ describes the left part
of a transition diagram, plus information on which vertices might end up being
identified by $f$ on the other side of the diagram. This is the only data
that can be preserved by an invertible multiplication on the right, and
conversely any element is reachable by multipying by a suitable planar partial
bijection.
  Similarly, the $\GrL$-class (right ideal generated by $f$) is determined by $f_\lefttorightarrow$ and a set
$f_l$ defined in a dual manner.
Then it can be observed that the resulting joint map
$f \mapsto (f_\lefttorightarrow, f_l, f_\righttoleftarrow, f_r)$ is injective
because of planarity. So in particular, every class $\GrH$-class (an intersection
of an $\GrR$ and a $\GrL$-class) is a singleton, which is equivalent to the monoid
being aperiodic.

Note that for a monoid of reversible transitions,
the problem can reduce to the aperiodicity of the Jones monoids.
In that setting the Green relations can be characterized using similar (although
slightly simpler) maps as shown in~\citet[Theorem 3.5, (iii-v)]{GreenRelKauffman}.
\end{remark}

In any case, we can deduce the first half of the main theorem.
\begin{corollary}
\label{cor:main-planarisfo}
Every planar 2DFT computes a first-order transduction.
\end{corollary}
\begin{proof}
  This reduces to \Cref{thm:tlmonisaperiodic}, since
the behaviour monoid $\behavmon(\cT)$ of a planar transducer $\cT$ with
directed set of states $(Q, \pol)$ ordered by $<$ will always be a submonoid of $\tlmon(Q, \pol, <)$.
\end{proof}

\section{Planar 2RFTs compute all first-order transductions}
\label{sec:krohnrhodes}

The objective of this section is to establish the other half of~\Cref{thm:main}.

\begin{theorem}
  Any first-order transduction can be computed by a planar 2RFT.
\end{theorem}

Let us walk through the proof of the theorem, delegating important lemmas to
further subsections.
We use two results from the literature to characterize first-order transductions
as compositions of simpler functions.
The first is \citeauthor{KrohnRhodes}'s \citeyearpar{KrohnRhodes} decomposition theorem,
specialized to the following standard statement found in e.g.\ \cite[Theorem~4.8]{aperiodic}
about \emph{sequential functions} -- these are the functions computed by \emph{one-way} deterministic transducers (i.e.\ 2DFTs with $\rho(Q) = \{1\}$).

\begin{theorem}[Aperiodic Krohn--Rhodes decomposition]
  \label{thm:kr}
  Any \emph{aperiodic sequential function} $f\colon \Sigma^* \to \Gamma^*$ can be
  realized as a composition $f = f_1 \circ \ldots \circ f_n$ (with $f_i :
  \Delta_i^* \to \Delta_{i-1}^*$, $\Delta_0 = \Gamma$ and $\Delta_n =
  \Sigma$) where each function $f_i$ is computed by some aperiodic sequential
  transducer \emph{with 2 states}.
\end{theorem}

We can show that planar 2RFTs compute all aperiodic sequential
functions by establishing that that the functions computed by planar
2RFTs are closed under composition (\Cref{lem:revPlanarCompo}), and by
providing an encoding of aperiodic 2-state transducers (\Cref{lem:flipflop}).

Then we use another result, which is derived from the characterization of
first-order transductions by aperiodic streaming string
transducers~\citep{FOSST,AperiodicSST}.
\begin{lemma}[{rephrasing of~\citep*[Lemma~4.8]{ListFunctions}\footnote{We use
      the theorem numbering from the official published version
      of~\citep{ListFunctions}, which is significantly different from the
      numbering in the arXiv version. Our rephrasing uses the fact that an
      aperiodic \emph{rational} function can be decomposed as $\ttreverse\circ
      g\circ\ttreverse\circ h$ where $g$ and $h$ are aperiodic sequential; this
      decomposition, analogous to a theorem of~\citet{ElgotMezei} on general rational functions, can be proved directly
      starting from the definition of rational function with monoids that is
      used by~\citet{ListFunctions}.}}]
  Every first-order transduction can be decomposed as $f \circ \ttreverse \circ g \circ
  \ttreverse \circ h$ where $f$ is computed by a \emph{monotone register
  transducer} and the functions $g$ and $h$ are aperiodic sequential.
\end{lemma}
The $\ttreverse$ function simply refers to a function that reverses its input
and can easily be implemented with a $3$-state reversible planar transducer.
Monotone register transducers are a variation of single-state streaming
stream transducers \citep{SST} with a monotonicity constraint on register assignments.
So we can conclude the proof using an encoding of monotone register transducers
in planar 2RFTs (\Cref{lem:mrt}) and re-using the closure under composition
of planar 2RFTs.

Our remaning task is to prove the aforementioned lemmas. The following subsections will be concerned with
\Cref{lem:revPlanarCompo},~\Cref{lem:flipflop} and~\Cref{lem:mrt} respectively.

\subsection{Closure under composition}
\label{subsec:revPlanarCompo}

The proof of \citep[Theorem~1]{ReversibleTransducers} gives an efficient
construction to compose 2RFTs reminiscent of the
wreath product of semigroups. The basic idea is that one may simply take a
cartesian product of the states, and a transition is composed of:
\begin{itemize}
  \item a first component that simulates the first transducer in the composition,
  \item and a
  second component that simulates the second transducer on the production given by the first
  component, signaling to the simulation of the first transducer to proceed
  with its computation or rewind according to its needs (it is because of
  rewinding that reversibility is required).
\end{itemize}

It turns out that this construction also readily preserves planar behaviours.
We illustrate this in~\Cref{fig:compEx} and thus can show the following.
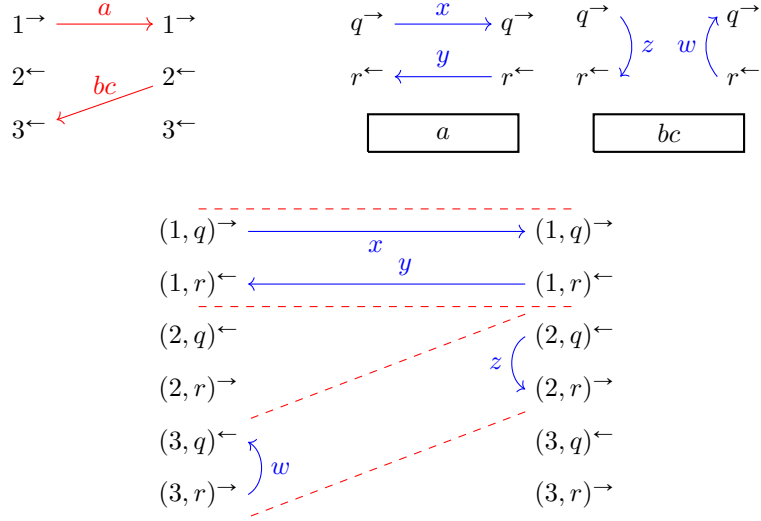
\begin{figure}
  \begin{center}
    \begin{tikzpicture}
      \node (11) at (0,2.7) {$1^\rightarrow$};
      \node (21) at (0,2) {$2^\leftarrow$};
      \node (31) at (0,1.3) {$3^\leftarrow$};
      \node (12) at (2,2.7) {$1^\rightarrow$};
      \node (22) at (2,2) {$2^\leftarrow$};
      \node (32) at (2,1.3) {$3^\leftarrow$};
      \draw[->, color=red] (11) -- node [above] {$a$} (12);
      \draw[->, color=red] (22) -- node [above] {$bc$} (31);

      \node (q1) at (4.5,2.7) {$q^\rightarrow$};
      \node (r1) at (4.5,2) {$r^\leftarrow$};
      \node (q2) at (6.5,2.7) {$q^\rightarrow$};
      \node (r2) at (6.5,2) {$r^\leftarrow$};
      \draw[->, color=blue] (q1) -- node [above] {$x$} (q2);
      \draw[->, color=blue] (r2) -- node [above] {$y$} (r1);

      \node (q1') at (7.5,2.8) {$q^\rightarrow$};
      \node (r1') at (7.5,2) {$r^\leftarrow$};
      \node (q2') at (9.5,2.8) {$q^\rightarrow$};
      \node (r2') at (9.5,2) {$r^\leftarrow$};
      \draw[->, color=blue] (q1'.east) to [bend left = 40] node [right] {$z$} (r1'.east);
      \draw[->, color=blue] (r2'.west) to [bend left = 40] node [left] {$w$} (q2'.west);

      \draw[thick] (4.5,1) -- (4.5,1.5) -- (6.5,1.5) -- (6.5,1) -- (4.5,1);
      \node at (5.5,1.25) {$a$};
      \draw[thick] (7.5,1) -- (7.5,1.5) -- (9.5,1.5) -- (9.5,1) -- (7.5,1);
      \node at (8.5,1.25) {$bc$};
    \end{tikzpicture}

\vspace{2em}
    \begin{tikzpicture}
      \node (1q1) at (0,3.5) {$(1,q)^\rightarrow$};
      \node (1r1) at (0,2.8) {$(1,r)^\leftarrow$};
      \node (2q1) at (0,2.1) {$(2,q)^\leftarrow$};
      \node (2r1) at (0,1.4) {$(2,r)^\rightarrow$};
      \node (3q1) at (0,0.7) {$(3,q)^\leftarrow$};
      \node (3r1) at (0,0) {$(3,r)^\rightarrow$};

      \node (1q2) at (5,3.5) {$(1,q)^\rightarrow$};
      \node (1r2) at (5,2.8) {$(1,r)^\leftarrow$};
      \node (2q2) at (5,2.1) {$(2,q)^\leftarrow$};
      \node (2r2) at (5,1.4) {$(2,r)^\rightarrow$};
      \node (3q2) at (5,0.7) {$(3,q)^\leftarrow$};
      \node (3r2) at (5,0) {$(3,r)^\rightarrow$};

      \draw[dashed, color=red] (1q1.north) -- (1q2.north);
      \draw[dashed, color=red] (1r1.south) -- (1r2.south);
      \draw[->, color=blue] (1q1) -- node [below] {\!\!\!\!\!$x$} (1q2);
      \draw[->, color=blue] (1r2) -- node [above] {\;\;\;\;\;$y$} (1r1);
      
      \draw[dashed, color=red] (2q2.north west) -- (3q1.north east);
      \draw[dashed, color=red] (2r2.south west) -- (3r1.south east);
      \draw[->, color=blue] (2q2.west) to [bend right = 60] node [left] {$z$} (2r2.west);
      \draw[->, color=blue] (3r1.east) to [bend right = 60] node [right] {$w$} (3q1.east);
    \end{tikzpicture}
\end{center}
  \caption{Three planar transitions, one with output in $\{a,b,c\}$, two with inputs
  in $\{a,b,c\}$, and the transition obtained by composing the transducers they
  originate from.}
  \label{fig:compEx}
\end{figure}

\begin{lemma}
  \label{lem:revPlanarCompo}
  Let $\cT_1 : \Sigma^* \partto \Gamma^*$ and $\cT_2 : \Gamma^* \partto \Delta^*$
  be planar 2RFTs. There is a
  planar 2RFT computing the composition of the partial
  functions $\cT_2 \circ \cT_1$.
\end{lemma}
\begin{proof}
If $\cT_1 = (Q, \rho_1, q_0, F_1, \delta_1)$ and
$\cT_2 = (R, \rho_2, r_0, F_2, \delta_2)$, the composition is computed by
\[ \cT_2 * \cT_1 = (Q \times R, \rho_1\rho_2, (q_0, r_0), F_1 \times F_2, \delta)\]
  with $\delta : \Sigma_\bumpers \to \powerset((Q \times R) \times \Delta^* \times (Q \times R)$
  defined as a union $\deltalr \cup \deltarl \cup
  \deltall \cup \deltarr$, where each component is the smallest such that
  \begin{itemize}[leftmargin=*]
\item
$\deltalr(a) \ni ((q,r),w,(q',r'))$ if 
$\rho_2(r)=\rho_2(r')=1$,
$(\varepsilon,r, v)
\xrightarrow{\smash{\; w\;}}^*_{\delta_2}
(v,r',\varepsilon)$ and
$(q,v,q') \in \delta_1'(a)$
%\cT_2 moves left-to-right
\item
$\deltarl(a) \ni ((q',r),w,(q,r'))$
if
$\rho_2(r) = \rho_2(r')=-1$,
$
(v,r, \varepsilon)
\xrightarrow{\smash{\; w\;}}^*_{\delta_2}
(\varepsilon,r',v)
$
and
$(q,v,q') \in \delta_1'(a)$
%\cT_2 moves right-to-left
\item
$\deltall(a) \ni ((q,r),w,(q,r'))$
if
$\rho_2(r) = -\rho_2(r')=1$,
$
(\varepsilon,r,v)
\xrightarrow{\smash{\; w\;}}^*_{\delta_2}
(\varepsilon,r',v)
$
and
$(q,v,q') \in \delta_1'(a)$
%\cT_2 moves left-to-left
\item
$\deltarr(a) \ni ((q',r),w,(q',r'))$
if
$\rho_2(r) =-\rho_2(r')=-1$,
$
(v,r,\varepsilon)
\xrightarrow{\smash{\; w\;}}^*_{\delta_2}
(v,r',\varepsilon)
$
and
$(q,v,q') \in \delta_1'(a)$
%\cT_2 moves right-to-right
\end{itemize}
where $\delta_1' : \Sigma \to \powerset(Q \times \Gamma_\bumpers^* \times Q)$
is the obvious modification of $\delta_1$ required to surround the output of
$\cT_1$ by the end markers $\triangleright$ and $\triangleleft$.
The basic idea is that in each of the four cases, we observe the global
  movement of $\cT_2$ on the output produced by $\cT_1$ (determined by the
subscripts $\rightarrow, \leftarrow, \lefttorightarrow$ or $\righttoleftarrow$),
we move forward or backwards in the run of $\cT_1$
over the input accordingly.

%Formally, $\delta_1'$ is defined by taking
%\begin{itemize}
%  \item $(q_0, \triangleright w, q) \in  \delta_1'(\triangleright)$ iff
%    $(q_0, w, q) \in \delta_1(a)$
%  \item $(q, w \triangleleft, q_f) \in \delta_1'(\triangleleft)$ iff $(q,w,q_f) \in \delta_1(\triangleright)$ and $q_f \in F$
%  \item $(q,w,q') \in \delta_1'(a)$ iff $(q, w, q') \in \delta_1(a)$ and we
%    don't have $q' \in F$ and $a = \triangleleft$ or $q = q_0$ and $a = \triangleright$.
%\end{itemize}

The construction is exactly the same as in the proof of \citep[Theorem~1]{ReversibleTransducers};
it is sound and preserves both determinism and reversibility for exactly the same reasons.
It remains to check that the obtained transitions can still be made planar.
It is the case: if $\delta_1$
and $\delta_2$ induce planar transitions with respect to two orders over $Q$ and
  $R$ respectively, then
$\delta$ induce planar transitions with respect to their lexicographic order over
$Q \times R$ as per~\Cref{def:planarity}.

To do so, let us assume that we have two edges
$e_0 = \{(q_0,r_0),p_0), ((q_2,r_2),p_2)\} \in
\cG(\rho_1\rho_2,\lfloor \delta\rfloor r(a))$
and
$e_1 = \{(q_1,r_1),p_1), ((q_3,r_3),p_3)\} \in
\cG(\rho_1\rho_2,\lfloor \delta \rfloor(a))$
such that $((q_i, r_i), p_i) < ((q_j,r_j),p_j)$ whenever $i < j$.
We then proceed via a case analysis:
\begin{itemize}
\item If it is the case that every $\rho_2(r_i)$ coincide, then projecting the
  $R$ components yield edges either in $\cG(\rho_1, \lfloor\delta_1\rfloor(a))$.
  In that case, we contradict the planarity of $\delta_1(a)$.
\item If three of the $q_i$s coincide, by determinism and
      reversibility of $\cT_1$, it means that projecting the $Q$ components
      away yields edges in $\cG(\rho_2, \lfloor\delta_2\rfloor(v))$ where
      $v$ is the output of $\cT_1$ when it reads $a$ from state $q$. But
      then we contradict the planarity of transitions in $\cT_2$.
\item And with that, we have covered all the cases; the first item allows us
  to show that either $e_0$ or $e_1$ is in
$\cG(\rho_1\rho_2,\lfloor\deltall \cup \deltarr\rfloor(a))$.
If $e_0$ is an edge of $\cG(\rho_1\rho_2,\lfloor\deltall \cup \deltarr\rfloor(a))$,
it means in particular that $q_0 = q_2$ and $\rho_2(r_0) \neq \rho_2(r_2)$.
So necessarily $p_0 = p_1 = p_2$, so we further have $q_0 = q_1 = q_2$, so we
were actually in the second case.
The case where $e_1$ is an edge of
$\cG(\rho_1\rho_2,\lfloor\deltall \cup \deltarr\rfloor(a))$ is handled similarly.
\end{itemize}
\end{proof}

\subsection{Two-state aperiodic sequential transducers}

We now turn to proving the following.

\begin{lemma}
  \label{lem:flipflop}
Every two-state aperiodic sequential transducer can be computed by a
planar 2RFT.
\end{lemma}

Before embarking on this proof, let us recall what is an aperiodic
sequential transducer.

\begin{definition}
  \label{def:seqTrans}
  A (total deterministic) \emph{sequential transducer} with input alphabet $\Sigma$ and output
alphabet $\Gamma$ is a tuple $(Q,q_0,\partial,o)$ where $Q$ is a finite
  set of states, $q_0 \in Q$ is an initial state, $o\colon Q \to \Gamma^*$ is a \emph{final output}
  function and $\partial : \Sigma \to Q \to Q \times \Gamma^*$ is a transition
  function. Let us write $\partial_1\colon \Sigma \to Q \to Q$ and
  $\partial_2\colon \Sigma \to Q \to \Gamma^*$ for each component of $\partial$.

The function $\partial$ is extended to $\partial^*\colon \Sigma^* \to Q \to Q \times \Gamma^*$
defined by its components $\partial_1^*$ and $\partial_2^*$:
\[\partial^*(\varepsilon)(q) = (q, \varepsilon) \qquad
  \partial_1^*(w a) = \partial_1(a) \circ \partial_1^*(w) \quad
  \text{and} \quad \partial_2^*(wa)(q) = \partial_2^*(w)\partial_2(a)\left(\partial_1^*(w)(q)\right)\]

Such a sequential transducer then defines a function $\Sigma^* \to \Gamma^*$
  by $w \mapsto \partial_2^*(w)\cdot o(\partial_1^*(q_0))$.

This sequential transducer is said to be \emph{aperiodic} if the monoid
generated by the transition function, which is the
closure of $\{ \partial_1(a) \mid a \in \Sigma\}$ under composition, is aperiodic.
\end{definition}

An example of a two-state aperiodic transducer is pictured in~\Cref{fig:2StApTr}.
One can see that the monoid generated by its transition function is actually
the (unique) maximal aperiodic submonoid of $\{1,2\}^{\{1,2\}}$ obtained
by removing the only non-trivial bijection.

We now have all the definitions required to embark in the construction that
will prove~\Cref{lem:flipflop}; before giving the details, one may want
to take a look at the translation of the transducer of~\Cref{fig:2StApTr},
which is ran over an input in~\Cref{fig:Ex2StApTr22DFT}.

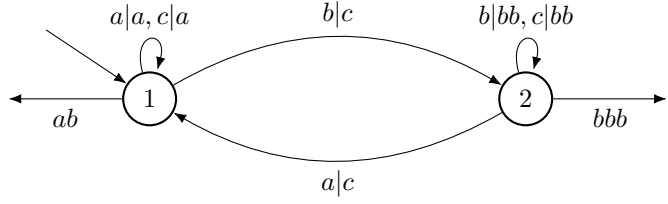
\begin{figure}
  \begin{center}
    \begin{tikzpicture}[>=Latex]
      \node (init) at (-0.5,2) {};
      \node (outa) at (-1,1) {};
      \node (outb) at (8,1) {};
      \node[circle,draw,thick,minimum size=20pt] (qa) at (1,1) {$1$};
      \node[circle,draw,thick,minimum size=20pt] (qb) at (6,1) {$2$};
      
      {\draw[->] (init) -- (qa);}
      {\draw[->] (qa) edge [loop above] node {$a|a, c|a$} (); }
      {\draw[->] (qa) edge  [bend left] node[above] {$b|c$} (qb);}
      {\draw[->] (qb) edge [loop above] node {$b|bb, c|bb$} ();}
      {\draw[->] (qb) edge [bend left] node[below] {$a|c$} (qa);}
      {\draw[->] (qa) edge node[below] {$ab$} (outa);}
      {\draw[->] (qb) edge node[below] {$bbb$}(outb);}
    \end{tikzpicture}
  \end{center}
  \caption{A two-state aperiodic sequential transducer over the alphabet
  $\{a,b,c\}$
  A transition $\partial(a)(q) = (q',w)$
  is pictured as an arrow from $q$ to $q'$ labelled by $a | w$.}
  \label{fig:2StApTr}
\end{figure}

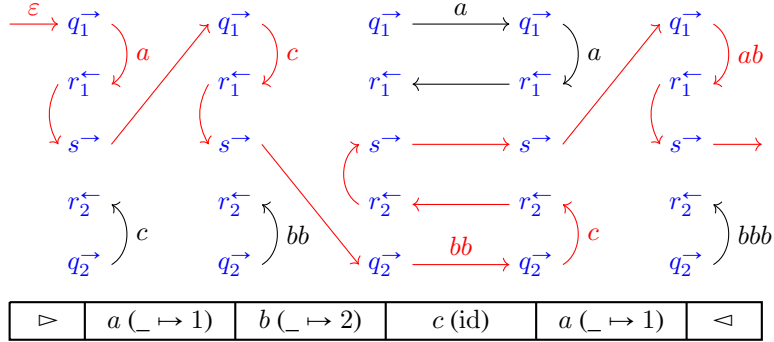
\begin{figure}
  \begin{center}
    \begin{tikzpicture}
      \draw[thick] (-1,0.5) -- (9,0.5) -- (9,0) -- (-1,0) -- (-1,0.5);
      \draw[thick] (0,0) -- (0,0.5);
      \draw[thick] (2,0) -- (2,0.5);
      \draw[thick] (4,0) -- (4,0.5);
      \draw[thick] (6,0) -- (6,0.5);
      \draw[thick] (8,0) -- (8,0.5);
      \node at (-0.5,0.25) {$\triangleright$};
      \node at (1,0.25) {$a$ ($\_ \mapsto 1$)};
      \node at (3,0.25) {$b$ ($\_ \mapsto 2$)};
      \node at (5,0.25) {$c$ ($\mathrm{id}$)};
      \node at (7,0.25) {$a$ ($\_ \mapsto 1$)};
      \node at (8.5,0.25) {$\triangleleft$};

      \node[color=blue] (q15) at (8,4.2) {$q_1^\rightarrow$};
      \node[color=blue] (r15) at (8,3.4) {$r_1^\leftarrow$};
      \node[color=blue] (s5) at (8,2.6) {$s^\rightarrow$};
      \node[color=blue] (r25) at (8,1.8) {$r_2^\leftarrow$};
      \node[color=blue] (q25)  at (8,1) {$q_2^\rightarrow$};

      {\color{red}\draw[->] (s5) -- (9,2.6);
      \draw[->] (q15.east) to [bend left = 60] node [right] {$ab$} (r15.east);}
      \draw[->] (q25.east) to [bend right = 60] node [right] {$bbb$} (r25.east);

        \node[color=blue] (q14) at (6,4.2) {$q_1^\rightarrow$};
        \node[color=blue] (r14) at (6,3.4) {$r_1^\leftarrow$};
        \node[color=blue] (s4) at (6,2.6) {$s^\rightarrow$};
        \node[color=blue] (r24) at (6,1.8) {$r_2^\leftarrow$};
        \node[color=blue] (q24)  at (6,1) {$q_2^\rightarrow$};

        {\color{red} \draw[->] (r15.west) to [bend right = 30] (s5.west);
        \draw[->] (s4.east) to (q15.west);
        \draw[->] (q24.east) to [bend right = 60] node [right] {$c$} (r24.east);}
        \draw[->] (q14.east) to [bend left = 60] node [right] {$a$} (r14.east);

        \node[color=blue] (q13) at (4,4.2) {$q_1^\rightarrow$};
        \node[color=blue] (r13) at (4,3.4) {$r_1^\leftarrow$};
        \node[color=blue] (s3) at (4,2.6) {$s^\rightarrow$};
        \node[color=blue] (r23) at (4,1.8) {$r_2^\leftarrow$};
        \node[color=blue] (q23)  at (4,1) {$q_2^\rightarrow$};

        \draw[->] (q13) -- node[above] {$a$} (q14);
        \draw[->] (r14) -- (r13);
        {\color{red}\draw[->] (q23) -- node[above] {$bb$} (q24);
        \draw[->] (r24) -- (r23);
        \draw[->] (s3) -- (s4);}

        \node[color=blue] (q12) at (2,4.2) {$q_1^\rightarrow$};
        \node[color=blue] (r12) at (2,3.4) {$r_1^\leftarrow$};
        \node[color=blue] (s2) at (2,2.6) {$s^\rightarrow$};
        \node[color=blue] (r22) at (2,1.8) {$r_2^\leftarrow$};
        \node[color=blue] (q22)  at (2,1) {$q_2^\rightarrow$};

        {\color{red}\draw[->] (r23.west) to [bend left = 60] (s3.west);
        \draw[->] (s2.east) to (q23.west);
        \draw[->] (q12.east) to [bend left = 60] node [right] {$c$} (r12.east);}
        \draw[->] (q22.east) to [bend right = 60] node [right] {$bb$} (r22.east);

        \node[color=blue] (q11) at (0,4.2) {$q_1^\rightarrow$};
        \node[color=blue] (r11) at (0,3.4) {$r_1^\leftarrow$};
        \node[color=blue] (s1) at (0,2.6) {$s^\rightarrow$};
        \node[color=blue] (r21) at (0,1.8) {$r_2^\leftarrow$};
        \node[color=blue] (q21)  at (0,1) {$q_2^\rightarrow$};

        \draw[->] (q21.east) to [bend right = 60] node [right] {$c$} (r21.east);
          {\color{red}\draw[->] (-1,4.2) -- node[above] {${\vareps}$} (q11);
            \draw[->] (q11.east) to [bend left = 60] node [right] {$a$} (r11.east);
            \draw[->] (r11.west) to [bend right = 30] (s1.west);
            \draw[->] (s1.east) to (q12.west);
            \draw[->] (r12.west) to [bend right = 30] (s2.west);}
      % }
    \end{tikzpicture}
  \end{center}
\caption{A run of the 2DFT obtained by the construction in the proof
of~\Cref{lem:flipflop} applied to the sequential transducer pictured
in~\Cref{fig:2StApTr}.}
\label{fig:Ex2StApTr22DFT}
\end{figure}

\begin{proof}[of~\Cref{lem:flipflop}]
Without loss of generality, suppose that we start with an input transducer
  of shape $(\{1,2\}, 1, \partial, o)$ with input alphabet $\Sigma$ and output
  alphabet $\Gamma$.
  We build a planar reversible 2DFT with the directed state-space
  $\{q_1^\rightarrow, r_1^\leftarrow, s^\rightarrow, r_2^\rightarrow, q_2^\rightarrow\}$,
  the initial state $q_1^\rightarrow$ and final state $s^\rightarrow$.
  Its transition function $\delta$ is then defined as per the table given
  in~\Cref{fig:flipFlopTransTrans}: note that for $a \in \Sigma$, $\delta(a)$
  is determined by $\partial(a)$, $\delta(\triangleright)$ is fixed (essentially
  because we fixed that $1$ was the initial state of our input sequential transducer)
  and $\delta(\triangleleft)$ is determined by $o$.

  As a visual inspection of \Cref{fig:flipFlopTransTrans} reveals, all those transitions are
  reversible and planar for the ordering $q_1^\rightarrow < r_1^\leftarrow < s^\rightarrow < r_2^\rightarrow < q_2^\rightarrow$ on the
  states.

  Now it only remains to be checked that the partial functions compted by the two
  machines are the same. The proof of this basically relies on the following
  claim that holds by induction over $w \in \Sigma^*$.
\begin{center}
  If $\partial^*(w)(1) = (i,u)$, then
  $(\varepsilon, q_1^\rightarrow, \triangleright w)
  \xrightarrow{\smash{\; u \;}}^*_{\delta}
  (\triangleright w,q_i^\rightarrow,\varepsilon)$ and
  $(\triangleright w,r_i^\leftarrow,\varepsilon)
  \xrightarrow{\smash{\; \varepsilon \;}}^*_{\delta}
(\triangleright w,s^\rightarrow,\varepsilon)$
\end{center}
  With that proven, we may build the following run over a word $w$
\[
  (\varepsilon, q_1^\rightarrow, \triangleright w \triangleleft)
  \xrightarrow{\smash{\; \partial_2^*(w)(1) \;}}^*_{\delta}
  (\triangleright w, q_i^\rightarrow, \triangleleft)
  \xrightarrow{\smash o(i)}_\delta
  (\triangleright w, r_i^\leftarrow, \triangleleft)
  \xrightarrow{\smash{\; \varepsilon \;}}^*_{\delta}
(\triangleright w,s^\rightarrow,\triangleleft)
  \xrightarrow{\smash{\; \varepsilon \;}}_{\delta}
(\triangleright w \triangleleft,s^\rightarrow,\varepsilon)
\]
  where $i = \partial_1^*(w)(1)$, which allows to conclude.
\end{proof}

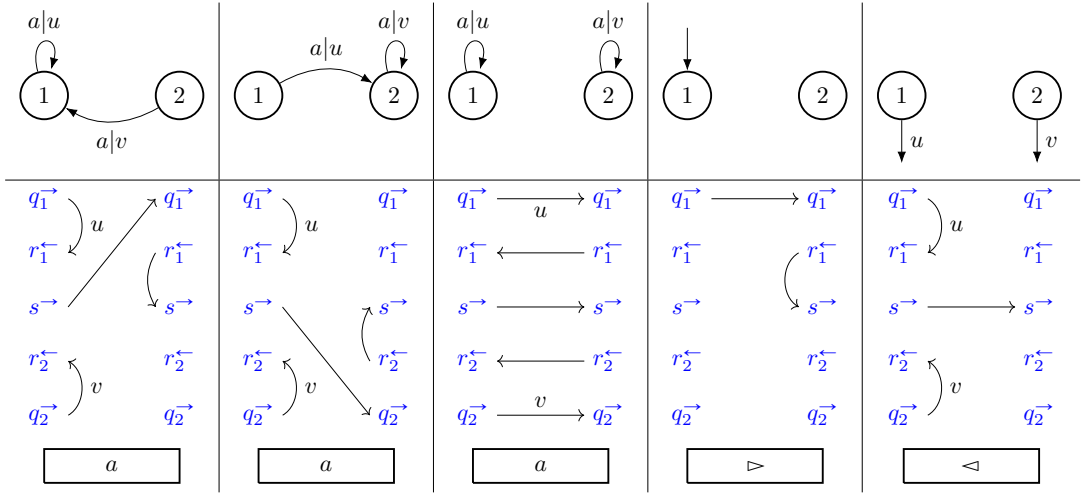
\begin{figure}
  \begin{center}
  \scalebox{0.9}{
  \begin{tabular}{c|c|c|c|c}
    \begin{tikzpicture}[>=Latex]
      \node[draw=none] () at (1,0) {};
      \node[circle,draw,thick,minimum size=20pt] (qa) at (1,1) {$1$};
      \node[circle,draw,thick,minimum size=20pt] (qb) at (3,1) {$2$};
      
      {\draw[->] (qa) edge [loop above] node {$a|u$} (); }
      {\draw[->] (qb) edge [bend left] node[below] {$a|v$} (qa);}
    \end{tikzpicture}
    &
    \begin{tikzpicture}[>=Latex]
      \node[draw=none] () at (1,0) {};
      \node[circle,draw,thick,minimum size=20pt] (qa) at (1,1) {$1$};
      \node[circle,draw,thick,minimum size=20pt] (qb) at (3,1) {$2$};
      
      {\draw[->] (qa) edge  [bend left] node[above] {$a|u$} (qb);}
      {\draw[->] (qb) edge [loop above] node {$a|v$} ();}
    \end{tikzpicture}
    &
    \begin{tikzpicture}[>=Latex]
      \node[draw=none] () at (1,0) {};
      \node[circle,draw,thick,minimum size=20pt] (qa) at (1,1) {$1$};
      \node[circle,draw,thick,minimum size=20pt] (qb) at (3,1) {$2$};
      
      {\draw[->] (qa) edge [loop above] node {$a|u$} (); }
      {\draw[->] (qb) edge [loop above] node {$a|v$} ();}
    \end{tikzpicture}
    &
    \begin{tikzpicture}[>=Latex]
      \node[draw=none] () at (1,0) {};
      \node[circle,draw,thick,minimum size=20pt] (qa) at (1,1) {$1$};
      \node[circle,draw,thick,minimum size=20pt] (qb) at (3,1) {$2$};

      {\draw[->] (1,2) to (qa);}     
    \end{tikzpicture}
    &
    \begin{tikzpicture}[>=Latex]
      \node[draw=none] () at (1,0) {};
      \node[circle,draw,thick,minimum size=20pt] (qa) at (1,1) {$1$};
      \node[circle,draw,thick,minimum size=20pt] (qb) at (3,1) {$2$};

      {\draw[->] (qa) to node[right] {$u$} (1,0);}
      {\draw[->] (qb) to node[right] {$v$} (3,0);}
    \end{tikzpicture}
\\
\hline
    \begin{tikzpicture}
        \node[color=blue] (q11) at (0,4.2) {$q_1^\rightarrow$};
        \node[color=blue] (r11) at (0,3.4) {$r_1^\leftarrow$};
        \node[color=blue] (s1) at (0,2.6) {$s^\rightarrow$};
        \node[color=blue] (r21) at (0,1.8) {$r_2^\leftarrow$};
        \node[color=blue] (q21)  at (0,1) {$q_2^\rightarrow$};

        \node[color=blue] (q12) at (2,4.2) {$q_1^\rightarrow$};
        \node[color=blue] (r12) at (2,3.4) {$r_1^\leftarrow$};
        \node[color=blue] (s2) at (2,2.6) {$s^\rightarrow$};
        \node[color=blue] (r22) at (2,1.8) {$r_2^\leftarrow$};
        \node[color=blue] (q22)  at (2,1) {$q_2^\rightarrow$};

        \draw[->] (q21.east) to [bend right = 60] node [right] {$v$} (r21.east);
            \draw[->] (q11.east) to [bend left = 60] node [right] {$u$} (r11.east);
            \draw[->] (s1.east) to (q12.west);
            \draw[->] (r12.west) to [bend right = 30] (s2.west);
      \draw[thick] (0,0) -- (0,0.5) -- (2,0.5) -- (2,0) -- (0,0);
      \node at (1,0.25) {$a$};
    \end{tikzpicture}
    &
    \begin{tikzpicture}
        \node[color=blue] (q11) at (0,4.2) {$q_1^\rightarrow$};
        \node[color=blue] (r11) at (0,3.4) {$r_1^\leftarrow$};
        \node[color=blue] (s1) at (0,2.6) {$s^\rightarrow$};
        \node[color=blue] (r21) at (0,1.8) {$r_2^\leftarrow$};
        \node[color=blue] (q21)  at (0,1) {$q_2^\rightarrow$};

        \node[color=blue] (q12) at (2,4.2) {$q_1^\rightarrow$};
        \node[color=blue] (r12) at (2,3.4) {$r_1^\leftarrow$};
        \node[color=blue] (s2) at (2,2.6) {$s^\rightarrow$};
        \node[color=blue] (r22) at (2,1.8) {$r_2^\leftarrow$};
        \node[color=blue] (q22)  at (2,1) {$q_2^\rightarrow$};

        \draw[->] (q21.east) to [bend right = 60] node [right] {$v$} (r21.east);
            \draw[->] (q11.east) to [bend left = 60] node [right] {$u$} (r11.east);
            \draw[->] (s1.east) to (q22.west);
            \draw[->] (r22.west) to [bend left = 30] (s2.west);
      \draw[thick] (0,0) -- (0,0.5) -- (2,0.5) -- (2,0) -- (0,0);
      \node at (1,0.25) {$a$};
    \end{tikzpicture}
&
    \begin{tikzpicture}
        \node[color=blue] (q11) at (0,4.2) {$q_1^\rightarrow$};
        \node[color=blue] (r11) at (0,3.4) {$r_1^\leftarrow$};
        \node[color=blue] (s1) at (0,2.6) {$s^\rightarrow$};
        \node[color=blue] (r21) at (0,1.8) {$r_2^\leftarrow$};
        \node[color=blue] (q21)  at (0,1) {$q_2^\rightarrow$};

        \node[color=blue] (q12) at (2,4.2) {$q_1^\rightarrow$};
        \node[color=blue] (r12) at (2,3.4) {$r_1^\leftarrow$};
        \node[color=blue] (s2) at (2,2.6) {$s^\rightarrow$};
        \node[color=blue] (r22) at (2,1.8) {$r_2^\leftarrow$};
        \node[color=blue] (q22)  at (2,1) {$q_2^\rightarrow$};

        \draw[->] (s1.east) to (s2.west);
        \draw[->] (q11.east) to node [below] {$u$} (q12.west);
        \draw[->] (q21.east) to node [above] {$v$} (q22.west);
        \draw[->] (r12.west) to (r11.east);
        \draw[->] (r22.west) to (r21.east);
      \draw[thick] (0,0) -- (0,0.5) -- (2,0.5) -- (2,0) -- (0,0);
      \node at (1,0.25) {$a$};
    \end{tikzpicture}
&
    \begin{tikzpicture}
        \node[color=blue] (q11) at (0,4.2) {$q_1^\rightarrow$};
        \node[color=blue] (r11) at (0,3.4) {$r_1^\leftarrow$};
        \node[color=blue] (s1) at (0,2.6) {$s^\rightarrow$};
        \node[color=blue] (r21) at (0,1.8) {$r_2^\leftarrow$};
        \node[color=blue] (q21)  at (0,1) {$q_2^\rightarrow$};

        \node[color=blue] (q12) at (2,4.2) {$q_1^\rightarrow$};
        \node[color=blue] (r12) at (2,3.4) {$r_1^\leftarrow$};
        \node[color=blue] (s2) at (2,2.6) {$s^\rightarrow$};
        \node[color=blue] (r22) at (2,1.8) {$r_2^\leftarrow$};
        \node[color=blue] (q22)  at (2,1) {$q_2^\rightarrow$};

        \draw[->] (q11.east) to (q12.west);
        \draw[->] (r12.west) to [bend right = 60] (s2.west);
      \draw[thick] (0,0) -- (0,0.5) -- (2,0.5) -- (2,0) -- (0,0);
      \node at (1,0.25) {$\triangleright$};
    \end{tikzpicture}
&
    \begin{tikzpicture}
        \node[color=blue] (q11) at (0,4.2) {$q_1^\rightarrow$};
        \node[color=blue] (r11) at (0,3.4) {$r_1^\leftarrow$};
        \node[color=blue] (s1) at (0,2.6) {$s^\rightarrow$};
        \node[color=blue] (r21) at (0,1.8) {$r_2^\leftarrow$};
        \node[color=blue] (q21)  at (0,1) {$q_2^\rightarrow$};

        \node[color=blue] (q12) at (2,4.2) {$q_1^\rightarrow$};
        \node[color=blue] (r12) at (2,3.4) {$r_1^\leftarrow$};
        \node[color=blue] (s2) at (2,2.6) {$s^\rightarrow$};
        \node[color=blue] (r22) at (2,1.8) {$r_2^\leftarrow$};
        \node[color=blue] (q22)  at (2,1) {$q_2^\rightarrow$};

        \draw[->] (s1.east) to (s2.west);
        \draw[->] (q21.east) to [bend right = 60] node [right] {$v$} (r21.east);
        \draw[->] (q11.east) to [bend left = 60] node [right] {$u$} (r11.east);
      \draw[thick] (0,0) -- (0,0.5) -- (2,0.5) -- (2,0) -- (0,0);
      \node at (1,0.25) {$\triangleleft$};
    \end{tikzpicture}

\\
  \end{tabular}}
  \end{center}
  \caption{Definition of the translation from an aperiodic sequential
  transducer to a transition function for a 2DFT computing the same function.
  The first row represents features of the sequential transducer and the
  second a corresponding transition in the obtained 2DFT. In the first three
  columns, transitions $\partial(a)$ are depicted in diagramatic form, then
  the initial state in the fourth column and the output function in the last
  column.}
  \label{fig:flipFlopTransTrans}
\end{figure}

\subsection{Monotone register transducers}

We now prove the final lemma required to conclude the section.
\begin{lemma}
\label{lem:mrt}
Any string-to-string function computed by a monotone register
transducer may be computed by a planar reversible 2DFT.
\end{lemma}

Before proving this, we first need to define what is a
\emph{monotone register transducer}. This will
correspond to a restricted subclass of
\emph{copyless streaming string transducers} \citep{SST}.
Those machines go through their inputs in a single left-to-right pass,
storing infixes of their outputs in registers that they may update by
performing concatenations of previously stored values and constants.
Here we will further impose that our machines have no control states,
that the output corresponds to a single register and that the register
updates satisfy a monotonicity condition in addition to being copyless.

\begin{definition}\label{def:update}
  A \emph{register update} for the set of registers $R$ over the alphabet
  $\Gamma$ is a map $\sigma\colon R\to(R\cup\Gamma)^*$ (we assume
  $R\cap\Gamma=\varnothing$). It induces the map $\sigma^\dagger\colon
  (\Gamma^*)^R \to (\Gamma^*)^R$ defined this way: the $r$-component of the
  output is obtained by replacing, in $\sigma(r)$, each register occurrence $r'$
  by the $r'$-component of the input.
\end{definition}
We shall use a suggestive notation using an assignment symbol for register
updates: we write
\[ \sigma =
  \begin{cases}
    X \setreg XbYc\\
    Y \setreg abZ\\
    Z \setreg c
  \end{cases} \quad\text{for}\qquad \sigma(X)= XbYc,\; \sigma(Y) = abZ,\;
  \sigma(Z) = c
\]
(with $R = \{X,Y\}$ and $a,b,c\in\Gamma$). For this update, we have
$(\sigma^\dagger(\overrightarrow{w}))_X = w_X b w_Y c$ for any
$\overrightarrow{w}\in(\Gamma^*)^R$.
Now to limit ourselves to regular first-order transductions, we need
to apply some restrictions to the kind of register updates we have.
\begin{definition}
  Let $R = \{r_1 < \dots < r_n\}$ be a finite and totally ordered set of
  registers. A register update $\sigma\colon R\to(R\cup\Gamma)^*$ is
  \emph{copyless monotone} when the list of register occurrences in the
  concatenation $\sigma_a(r_1) \dots \sigma_a(r_n)$ is increasing (i.e.\
  strictly monotone).
\end{definition}
Our previous example of register update is copyless monotone for $X < Y < Z$:
the occurrences of registers in $XbYc \cdot abZ \cdot c$ form the increasing
sequence $X,Y,Z$. Some examples of register updates that do not satisfy this
condition include
\[ X \setreg XX \qquad
  \begin{cases}
    X \setreg Y\\
    Y \setreg X
  \end{cases}
\]
\begin{definition}
  A \emph{monotone register transducer} $\Sigma^*\to\Gamma^*$
  consists of the following:
  \begin{itemize}
  \item A finite totally ordered set of registers $R$.
  \item For each input letter $a\in\Sigma$, a copyless monotone register update
    $\sigma_a$ for $R$ over $\Gamma$.
  \end{itemize}
  It computes the function $a_1 \dots a_n \in \Gamma^* \mapsto
  (\sigma_{a_n}^\dagger \circ \dots \circ
  \sigma_{a_1}^\dagger(\overrightarrow\vareps))_X$ where
  $\overrightarrow\vareps\in(\Gamma^*)^R$ is the tuple whose components are
  all equal to the empty word, and $X$ is the minimum element of $R$.
\end{definition}

We can now turn to proving that they can be simulated by planar reversible
transducers; this will be a simpler variation of the construction found in the
proof of~\cite[Theorem~6]{AperiodicSST}
based on their notion of \emph{output structure} that determines how to
translate register updates into two-way behaviours; an example of applying
it is given in~\Cref{fig:mrt2dft}.

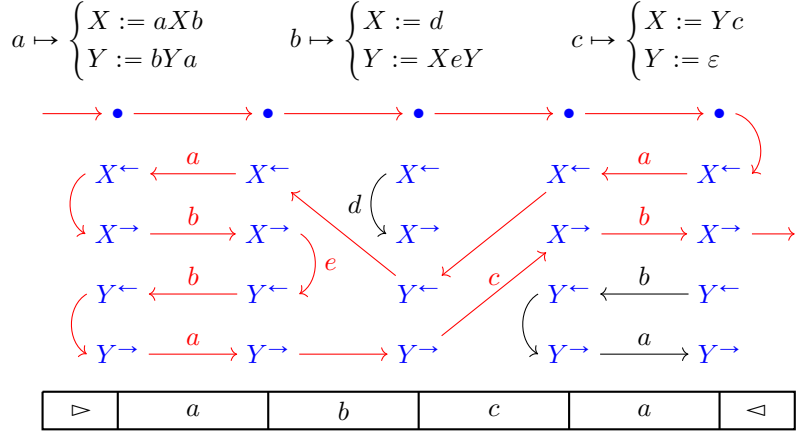
\begin{figure}
\[
  a \mapsto \begin{cases}
              X \setreg aXb\\
              Y \setreg bYa
            \end{cases} \qquad
  b \mapsto \begin{cases}
               X \setreg d \\
               Y \setreg XeY
            \end{cases} \qquad
  c \mapsto \begin{cases}
              X \setreg Yc \\
              Y \setreg \vareps
            \end{cases} \qquad
\]

  \begin{center}
    \begin{tikzpicture}
      \draw[thick] (-1,0.5) -- (9,0.5) -- (9,0) -- (-1,0) -- (-1,0.5);
      \draw[thick] (0,0) -- (0,0.5);
      \draw[thick] (2,0) -- (2,0.5);
      \draw[thick] (4,0) -- (4,0.5);
      \draw[thick] (6,0) -- (6,0.5);
      \draw[thick] (8,0) -- (8,0.5);
      \node at (-0.5,0.25) {$\triangleright$};
      \node at (1,0.25) {$a$};
      \node at (3,0.25) {$b$};
      \node at (5,0.25) {$c$};
      \node at (7,0.25) {$a$};
      \node at (8.5,0.25) {$\triangleleft$};

       \node[color=blue] (B5) at (8, 4.2) {$\bullet$};
      \node[color=blue] (Xi5) at (8,3.4) {$X^\leftarrow$};
      \node[color=blue] (Xo5) at (8,2.6) {$X^\rightarrow$};
      \node[color=blue] (Yi5) at (8,1.8) {$Y^\leftarrow$};
      \node[color=blue] (Yo5)  at (8,1) {$Y^\rightarrow$};

      {\color{red}\draw[->] (Xo5.east) to (9,2.6);
      \draw[->] (B5.east) to [bend left = 60] (Xi5.east);}

        \node[color=blue] (B4) at (6, 4.2) {$\bullet$};
        \node[color=blue] (Xi4) at (6,3.4) {$X^\leftarrow$};
        \node[color=blue] (Xo4) at (6,2.6) {$X^\rightarrow$};
        \node[color=blue] (Yi4) at (6,1.8) {$Y^\leftarrow$};
        \node[color=blue] (Yo4)  at (6,1) {$Y^\rightarrow$};

        {\color{red}\draw[->] (Xi5) -- node[above] {$a$} (Xi4);
        \draw[->] (Xo4) -- node[above] {$b$} (Xo5);
        \draw[->] (B4) -- (B5);}
        \draw[->] (Yi5) -- node[above] {$b$} (Yi4);
        \draw[->] (Yo4) -- node[above] {$a$} (Yo5);

        \node[color=blue] (B3) at (4, 4.2) {$\bullet$};
        \node[color=blue] (Xi3) at (4,3.4) {$X^\leftarrow$};
        \node[color=blue] (Xo3) at (4,2.6) {$X^\rightarrow$};
        \node[color=blue] (Yi3) at (4,1.8) {$Y^\leftarrow$};
        \node[color=blue] (Yo3)  at (4,1) {$Y^\rightarrow$};

        {\color{red}\draw[->] (Xi4) -- (Yi3);
        \draw[->] (Yo3) -- node[above] {$c$} (Xo4);
        \draw[->] (B3) -- (B4);
        }
        \draw[->] (Yi4.west) to [bend right = 60] (Yo4.west);

        \node[color=blue] (B2) at (2, 4.2) {$\bullet$};
        \node[color=blue] (Xi2) at (2,3.4) {$X^\leftarrow$};
        \node[color=blue] (Xo2) at (2,2.6) {$X^\rightarrow$};
        \node[color=blue] (Yi2) at (2,1.8) {$Y^\leftarrow$};
        \node[color=blue] (Yo2)  at (2,1) {$Y^\rightarrow$};

        {\color{red}
        \draw[->] (B2) -- (B3);
        \draw[->] (Yi3) -- (Xi2);
        \draw[->] (Xo2.east) to [bend left = 60] node[right] {$e$} (Yi2.east);
        \draw[->] (Yo2) -- (Yo3);
        }
        \draw[->] (Xi3.west) to [bend right = 60] node[left] {$d$} (Xo3.west);

        {\color{red}
        \node[color=blue] (B1) at (0, 4.2) {$\bullet$};
        \node[color=blue] (Xi1) at (0,3.4) {$X^\leftarrow$};
        \node[color=blue] (Xo1) at (0,2.6) {$X^\rightarrow$};}
        \node[color=blue] (Yi1) at (0,1.8) {$Y^\leftarrow$};
        \node[color=blue] (Yo1)  at (0,1) {$Y^\rightarrow$};

        {\color{red}
        \draw[->] (B1) -- (B2);
        \draw[->] (Xi2) -- node[above] {$a$} (Xi1);
        \draw[->] (Xo1) -- node[above] {$b$} (Xo2);
        \draw[->] (Yi2) -- node[above] {$b$} (Yi1);
        \draw[->] (Yo1) -- node[above] {$a$} (Yo2);}

        {\color{red}\draw[->] (-1,4.2) -- (B1);
        \draw[->] (Xi1.west) to [bend right = 60] (Xo1.west);
        \draw[->] (Yi1.west) to [bend right = 60] (Yo1.west);
        }
    \end{tikzpicture}
  \end{center}
\caption{A monotone register transducer and a run of the
corresponding planar reversible 2DFT (where we write $X^\leftarrow$ for $(X,-1)$ and
$X^\rightarrow$ for $(X,1)$ for readability).}
\label{fig:mrt2dft}
\end{figure}

\begin{proof}[of \Cref{lem:mrt}]
Given a monotone register transducer with an ordered set of registers $R$ and
  update function $\sigma$ and input alphabet $\Sigma$,
we will build a 2DFT with state space $\{\bullet\} \cup R \times \{-1,1\}$
(for some $\bullet \not\in R \times \{-1,1\}$).
The direction $\pol$ over states is going to be given by the second projection and
$\pol(\bullet) = 1$.
Its transition relation $\delta$ will be the smallest relation
satisfying the following, where $r_0$ is the smallest register of $R$:
\begin{itemize}
  \item $\{(\bullet, \varepsilon, \bullet), ((r,-1), \varepsilon, (r,1))\} \subseteq \delta(\triangleright)$
  and $\{(\bullet,\varepsilon, (r_0,-1)), ((r_0,1), \varepsilon, (r_0,1))\} \subseteq \delta(\triangleleft)$
  \item $((r,-1), \sigma(a), (r,1)) \in \delta(a)$ if $a \in \Sigma$ and $\sigma_a \in \Gamma^*$
  \item $((r,-1), w, (r',-1)) \in \delta(a)$ if $a \in \Sigma$ and $wr' \in \Gamma^*R$ is a prefix of $\sigma_a(r)$
  \item $((r,1), w, (r',-1)) \in \delta(a)$ if $a \in \Sigma$ and $rwr' \in R\Gamma^*R$ is an infix of $\sigma_a(r'')$ for some $r''$
  \item $((r,1), w, (r',1)) \in \delta(a)$ if $a \in \Sigma$ and $rw \in R\Gamma^*$ is a suffix of $\sigma_a(r')$
  \item $(\bullet, \varepsilon, \bullet) \in \delta(a)$ if $a \in \Sigma$
\end{itemize}

The basic idea is that we treat an entry in the state $(r,-1)$ as a query of what
is the value stored in the register $r$, which is going to be computed along
the run that will eventually loop back to state $(r, 1)$. Each case in
the definition above corresponds to producing part of the output described
by the register updates and querying previous register values accordingly.
Then to initialize the run, we simpy need to move to the rightmost position
and get into state $(r_0, -1)$; this is achieved by starting from the state
$\bullet$, which we take as the initial state. Then we set $(r_0, 1)$ to be
the unique final state.

Due to the fact that the register updates are copyless, we obtain that
the transitions are reversible and deterministic. The soundness of the
construction is easily derived by the invariant explained above, which can
be proven by induction.

As for planarity, the order on states will simply correspond to the lexicographic
  order over $R \times \{-1,1\}$, supplemented by stating that $\bullet$ is
  below all the other states. The monotonicity constraint on register updates
  then clearly ensures planarity of all transitions $\delta(a)$ for $a \in \Sigma$,
  and the extremal transitions $\delta(\triangleright)$
  (which only relates successive states in the order except for $\bullet$)
  and $\delta(\triangleleft)$ (that has two elements)
  are easily seen to be planar as well.
\end{proof}

\section{Some perspectives}

We hope to have demonstrated that the notions of planar deterministic/reversible automaton/transducer is an interesting
new way to capture star-free languages and first-order transductions.

As we explained in the \enquote{related work} section, this notion first arose~\citep{HinesPlanar} from the geometry of interaction semantics of non-commutative linear logic. Our interest in this semantics is part of a broader plan to link fragments
of the (linear) $\lambda$-calculus and automata theory that is exposed in~\citep{titoPhD}.
More specifically, we aim to improve the results of~\citep{aperiodic} -- characterizing star-free languages by means of a non-commutative $\lambda$-calculus --
to include first-order regular transductions using semantic methods.

Something worth investigating might be to extend the planarity condition to a
suitable notion of \emph{tree-walking automata} \citep[see e.g.][]{bojanczyk2008tree} and
% natural extensions to
 tree-walking
transducers outputting strings or trees. One natural question for instance
would be to ask whether reversible tree-to-tree walking transducers are also closed
under composition, and whether planarity would still be retained.

\paragraph{Acknowledgments} We would like to thank Sam van Gool, Denis Kuperberg, Paul-André Melliès and Noam Zeilberger for stimulating feedback on this work.

\bibliographystyle{plainnat}
% use the following instead if you encounter problems 
%\bibliographystyle{alpha}
\bibliography{bi}

\begin{thebibliography}{34}
\providecommand{\natexlab}[1]{#1}
\providecommand{\url}[1]{\texttt{#1}}
\expandafter\ifx\csname urlstyle\endcsname\relax
  \providecommand{\doi}[1]{doi: #1}\else
  \providecommand{\doi}{doi: \begingroup \urlstyle{rm}\Url}\fi

\bibitem[Abramsky(2007)]{AbramskyTemperleyLieb}
Samson Abramsky.
\newblock Temperley–{Lieb} {Algebra}: {From} {Knot} {Theory} to {Logic} and
  {Computation} via {Quantum} {Mechanics}.
\newblock In Goong Chen, Louis Kauffman, and Samuel Lomonaco, editors,
  \emph{Mathematics of {Quantum} {Computation} and {Quantum} {Technology}},
  volume 20074453, pages 515--558. Chapman and Hall/CRC, September 2007.
\newblock ISBN 978-1-58488-899-4 978-1-58488-900-7.
\newblock \doi{10.1201/9781584889007.ch15}.

\bibitem[Alur and {\v{C}}ern{\'{y}}(2010)]{SST}
Rajeev Alur and Pavol {\v{C}}ern{\'{y}}.
\newblock Expressiveness of streaming string transducers.
\newblock In Kamal Lodaya and Meena Mahajan, editors, \emph{{IARCS} Annual
  Conference on Foundations of Software Technology and Theoretical Computer
  Science, {FSTTCS} 2010, December 15-18, 2010, Chennai, India}, volume~8 of
  \emph{LIPIcs}, pages 1--12. Schloss Dagstuhl - Leibniz-Zentrum f{\"{u}}r
  Informatik, 2010.
\newblock \doi{10.4230/LIPIcs.FSTTCS.2010.1}.

\bibitem[Ananichev and Volkov(2004)]{SynchroMonotonic}
Dimitry~S. Ananichev and Mikhail~V. Volkov.
\newblock Synchronizing monotonic automata.
\newblock \emph{Theoretical Computer Science}, 327\penalty0 (3):\penalty0
  225--239, 2004.
\newblock \doi{10.1016/j.tcs.2004.03.068}.

\bibitem[Auinger(2014)]{Auinger14}
Karl Auinger.
\newblock Pseudovarieties generated by {Brauer} type monoids.
\newblock \emph{Forum Mathematicum}, 26\penalty0 (1):\penalty0 1--24, 2014.
\newblock \doi{doi:10.1515/form.2011.146}.

\bibitem[Baudru et~al.(2022)Baudru, Dando, Lhote, Monmege, Reynier, and
  Talbot]{PreRational}
Nicolas Baudru, Louis{-}Marie Dando, Nathan Lhote, Benjamin Monmege,
  Pierre{-}Alain Reynier, and Jean{-}Marc Talbot.
\newblock Weighted automata and expressions over pre-rational monoids.
\newblock In Florin Manea and Alex Simpson, editors, \emph{30th {EACSL} Annual
  Conference on Computer Science Logic, {CSL} 2022, February 14-19, 2022,
  G{\"{o}}ttingen, Germany (Virtual Conference)}, volume 216 of \emph{LIPIcs},
  pages 6:1--6:16. Schloss Dagstuhl - Leibniz-Zentrum f{\"{u}}r Informatik,
  2022.
\newblock \doi{10.4230/LIPIcs.CSL.2022.6}.

\bibitem[Birget(1989)]{Birget}
Jean{-}Camille Birget.
\newblock Concatenation of inputs in a two-way automaton.
\newblock \emph{Theoretical Computer Science}, 63\penalty0 (2):\penalty0
  141--156, 1989.
\newblock \doi{10.1016/0304-3975(89)90075-3}.

\bibitem[Bojańczyk and Colcombet(2008)]{bojanczyk2008tree}
Mikołaj Bojańczyk and Thomas Colcombet.
\newblock Tree-walking automata do not recognize all regular languages.
\newblock \emph{SIAM Journal on Computing}, 38\penalty0 (2):\penalty0 658--701,
  2008.
\newblock \doi{10.1137/050645427}.

\bibitem[Bojańczyk et~al.(2018)Bojańczyk, Daviaud, and
  Krishna]{ListFunctions}
Mikołaj Bojańczyk, Laure Daviaud, and Shankara~Narayanan Krishna.
\newblock Regular and {First}-{Order} {List} {Functions}.
\newblock In \emph{Proceedings of the 33rd {Annual} {ACM}/{IEEE} {Symposium} on
  {Logic} in {Computer} {Science} - {LICS} '18}, pages 125--134, Oxford, United
  Kingdom, 2018. ACM Press.
\newblock ISBN 978-1-4503-5583-4.
\newblock \doi{10.1145/3209108.3209163}.

\bibitem[Bonfante and Deloup(2019)]{BonfanteD19}
Guillaume Bonfante and Florian~L. Deloup.
\newblock Decidability of regular language genus computation.
\newblock \emph{Mathematical Structures in Computer Science}, 29\penalty0
  (9):\penalty0 1428--1443, 2019.
\newblock \doi{10.1017/S0960129519000057}.

\bibitem[Book and Chandra(1976)]{inhplanaraut}
Ronald~V. Book and Ashok~K. Chandra.
\newblock Inherently nonplanar automata.
\newblock \emph{Acta Informatica}, 6:\penalty0 89--94, 1976.
\newblock \doi{10.1007/BF00263745}.

\bibitem[Carton and Dartois(2015)]{cartondartoisaperiodic}
Olivier Carton and Luc Dartois.
\newblock {Aperiodic Two-way Transducers and FO-Transductions}.
\newblock In Stephan Kreutzer, editor, \emph{24th {EACSL} Annual Conference on
  Computer Science Logic, {CSL} 2015, September 7-10, 2015, Berlin, Germany},
  volume~41 of \emph{LIPIcs}, pages 160--174. Schloss Dagstuhl -
  Leibniz-Zentrum f{\"{u}}r Informatik, 2015.
\newblock \doi{10.4230/LIPIcs.CSL.2015.160}.

\bibitem[Clairambault and Murawski(2019)]{MAHORS}
Pierre Clairambault and Andrzej~S. Murawski.
\newblock On the {Expressivity} of {Linear} {Recursion} {Schemes}.
\newblock In Peter Rossmanith, Pinar Heggernes, and Joost-Pieter Katoen,
  editors, \emph{44th {International} {Symposium} on {Mathematical}
  {Foundations} of {Computer} {Science} ({MFCS} 2019)}, volume 138 of
  \emph{Leibniz {International} {Proceedings} in {Informatics} ({LIPIcs})},
  pages 50:1--50:14. Schloss Dagstuhl–Leibniz-Zentrum fuer Informatik, 2019.
\newblock ISBN 978-3-95977-117-7.
\newblock \doi{10.4230/LIPIcs.MFCS.2019.50}.

\bibitem[Dartois et~al.(2017)Dartois, Fournier, Jecker, and
  Lhote]{ReversibleTransducers}
Luc Dartois, Paulin Fournier, Ismaël Jecker, and Nathan Lhote.
\newblock On reversible transducers.
\newblock In Ioannis Chatzigiannakis, Piotr Indyk, Fabian Kuhn, and Anca
  Muscholl, editors, \emph{44th International Colloquium on Automata,
  Languages, and Programming, {ICALP} 2017, July 10-14, 2017, Warsaw, Poland},
  volume~80 of \emph{LIPIcs}, pages 113:1--113:12. Schloss Dagstuhl -
  Leibniz-Zentrum f{\"{u}}r Informatik, 2017.
\newblock \doi{10.4230/LIPIcs.ICALP.2017.113}.

\bibitem[Dartois et~al.(2018)Dartois, Jecker, and Reynier]{AperiodicSST}
Luc Dartois, Ismaël Jecker, and Pierre-Alain Reynier.
\newblock Aperiodic {String} {Transducers}.
\newblock \emph{International Journal of Foundations of Computer Science},
  29\penalty0 (05):\penalty0 801--824, August 2018.
\newblock ISSN 0129-0541, 1793-6373.
\newblock \doi{10.1142/S0129054118420054}.

\bibitem[East(2021)]{PresentationsTL}
James East.
\newblock {Presentations for Temperley–Lieb Algebras}.
\newblock \emph{The Quarterly Journal of Mathematics}, 72\penalty0
  (4):\penalty0 1253--1269, 02 2021.
\newblock ISSN 0033-5606.
\newblock \doi{10.1093/qmath/haab001}.

\bibitem[Elgot and Mezei(1965)]{ElgotMezei}
Calvin~C. Elgot and Jorge~E. Mezei.
\newblock On relations defined by generalized finite automata.
\newblock \emph{{IBM} Journal of Research and Development}, 9\penalty0
  (1):\penalty0 47--68, 1965.
\newblock \doi{10.1147/rd.91.0047}.

\bibitem[Fernau et~al.(2022)Fernau, Haase, and Hoffmann]{SynchroGame}
Henning Fernau, Carolina Haase, and Stefan Hoffmann.
\newblock The synchronization game on subclasses of automata.
\newblock In Pierre Fraigniaud and Yushi Uno, editors, \emph{11th International
  Conference on Fun with Algorithms, {FUN} 2022, May 30 to June 3, 2022, Island
  of Favignana, Sicily, Italy}, volume 226 of \emph{LIPIcs}, pages 14:1--14:17.
  Schloss Dagstuhl - Leibniz-Zentrum f{\"{u}}r Informatik, 2022.
\newblock \doi{10.4230/LIPIcs.FUN.2022.14}.

\bibitem[Filiot et~al.(2014)Filiot, Krishna, and Trivedi]{FOSST}
Emmanuel Filiot, Shankara~Narayanan Krishna, and Ashutosh Trivedi.
\newblock First-order definable string transformations.
\newblock In Venkatesh Raman and S.~P. Suresh, editors, \emph{34th
  International Conference on Foundation of Software Technology and Theoretical
  Computer Science, {FSTTCS} 2014, December 15-17, 2014, New Delhi, India},
  volume~29 of \emph{LIPIcs}, pages 147--159. Schloss Dagstuhl -
  Leibniz-Zentrum f{\"{u}}r Informatik, 2014.
\newblock ISBN 978-3-939897-77-4.
\newblock \doi{10.4230/LIPIcs.FSTTCS.2014.147}.

\bibitem[Goubault{-}Larrecq(2011)]{GoubaultLarrecqGoI}
Jean Goubault{-}Larrecq.
\newblock Musings around the geometry of interaction, and coherence.
\newblock \emph{Theoretical Computer Science}, 412\penalty0 (20):\penalty0
  1998--2014, 2011.
\newblock \doi{10.1016/j.tcs.2010.12.023}.

\bibitem[Higgins(1995)]{Higgins95}
Peter~M. Higgins.
\newblock Divisors of semigroups of order-preserving mappings on a finite
  chain.
\newblock \emph{International Journal of Algebra and Computation}, 5\penalty0
  (6):\penalty0 725, 1995.
\newblock \doi{10.1142/S0218196795000306}.

\bibitem[Hines(2003)]{Hines}
Peter Hines.
\newblock A categorical framework for finite state machines.
\newblock \emph{Mathematical Structures in Computer Science}, 13\penalty0
  (3):\penalty0 451--480, 2003.
\newblock \doi{10.1017/S0960129503003931}.

\bibitem[Hines(2006)]{HinesPlanar}
Peter Hines.
\newblock Temperley-{L}ieb {A}lgebras as two-way automata.
\newblock \url{http://www.dcs.gla.ac.uk/~simon/qnet/talks/Hines.pdf}, 2006.
\newblock Slides of a talk given at the QNET Workshop 2006.

\bibitem[Katsumata(2008)]{Katsumata08}
Shin{-}ya Katsumata.
\newblock Attribute grammars and categorical semantics.
\newblock In Luca Aceto, Ivan Damg{\aa}rd, Leslie~Ann Goldberg, Magn{\'{u}}s~M.
  Halld{\'{o}}rsson, Anna Ing{\'{o}}lfsd{\'{o}}ttir, and Igor Walukiewicz,
  editors, \emph{Automata, Languages and Programming, 35th International
  Colloquium, {ICALP} 2008, Reykjavik, Iceland, July 7-11, 2008, Proceedings,
  Part {II} - Track {B:} Logic, Semantics, and Theory of Programming {\&} Track
  {C:} Security and Cryptography Foundations}, volume 5126 of \emph{Lecture
  Notes in Computer Science}, pages 271--282. Springer, 2008.
\newblock \doi{10.1007/978-3-540-70583-3_23}.

\bibitem[Kauffman(1990)]{Kauffman}
Louis Kauffman.
\newblock An invariant of regular isotopy.
\newblock \emph{Transactions of the American Mathematical Society},
  318\penalty0 (2):\penalty0 417--471, April 1990.
\newblock \doi{10.2307/2001315}.

\bibitem[Krohn and Rhodes(1965)]{KrohnRhodes}
Kenneth Krohn and John Rhodes.
\newblock Algebraic theory of machines. {I}. {Prime} decomposition theorem for
  finite semigroups and machines.
\newblock \emph{Transactions of the American Mathematical Society},
  116:\penalty0 450--464, 1965.
\newblock ISSN 0002-9947, 1088-6850.
\newblock \doi{10.1090/S0002-9947-1965-0188316-1}.

\bibitem[Lau and FitzGerald(2006)]{GreenRelKauffman}
Kwok~Wai Lau and D.~G. FitzGerald.
\newblock Ideal structure of the kauffman and related monoids.
\newblock \emph{Communications in Algebra}, 34\penalty0 (7):\penalty0
  2617--2629, 2006.
\newblock \doi{10.1080/00927870600651414}.

\bibitem[Martynova and Okhotin(2023)]{ReversibleGraphWalking}
Olga Martynova and Alexander Okhotin.
\newblock State complexity of transforming graph-walking automata to halting,
  returning and reversible.
\newblock \emph{Information and Computation}, 291:\penalty0 105011, 2023.
\newblock \doi{10.1016/j.ic.2023.105011}.

\bibitem[Mohar and Thomassen(2001)]{MoharThomassen}
Bojan Mohar and Carsten Thomassen.
\newblock \emph{Graphs on Surfaces}.
\newblock Johns Hopkins series in the mathematical sciences. Johns Hopkins
  University Press, 2001.
\newblock ISBN 978-0-8018-6689-0.
\newblock \doi{10.56021/9780801866890}.

\bibitem[Muscholl and Puppis(2019)]{MuschollPuppis}
Anca Muscholl and Gabriele Puppis.
\newblock The {Many} {Facets} of {String} {Transducers}.
\newblock In Rolf Niedermeier and Christophe Paul, editors, \emph{36th
  {International} {Symposium} on {Theoretical} {Aspects} of {Computer}
  {Science} ({STACS} 2019)}, volume 126 of \emph{Leibniz {International}
  {Proceedings} in {Informatics} ({LIPIcs})}, pages 2:1--2:21. Schloss
  Dagstuhl–Leibniz-Zentrum fuer Informatik, 2019.
\newblock ISBN 978-3-95977-100-9.
\newblock \doi{10.4230/LIPIcs.STACS.2019.2}.

\bibitem[Nguy\~{\^e}n(2021)]{titoPhD}
Lê Thành~Dũng Nguy\~{\^e}n.
\newblock \emph{Implicit automata in linear logic and categorical transducer
  theory}.
\newblock PhD thesis, Université Paris XIII (Sorbonne Paris Nord), December
  2021.
\newblock URL \url{https://theses.hal.science/tel-04132636}.

\bibitem[Nguy{\~{ê}}n and Pradic(2020)]{aperiodic}
Lê Thành D\~ung Nguy{\~{ê}}n and Cécilia Pradic.
\newblock Implicit automata in typed $\lambda$-calculi {I}: aperiodicity in a
  non-commutative logic.
\newblock In Artur Czumaj, Anuj Dawar, and Emanuela Merelli, editors,
  \emph{47th International Colloquium on Automata, Languages, and Programming,
  {ICALP} 2020, July 8-11, 2020, Saarbr{\"{u}}cken, Germany (Virtual
  Conference)}, volume 168 of \emph{LIPIcs}, pages 135:1--135:20. Schloss
  Dagstuhl - Leibniz-Zentrum f{\"{u}}r Informatik, 2020.
\newblock \doi{10.4230/LIPIcs.ICALP.2020.135}.

\bibitem[Rabin and Scott(1959)]{RabinScott}
Michael~O. Rabin and Dana~S. Scott.
\newblock Finite automata and their decision problems.
\newblock \emph{{IBM} Journal of Research and Development}, 3\penalty0
  (2):\penalty0 114--125, 1959.
\newblock \doi{10.1147/rd.32.0114}.

\bibitem[Shepherdson(1959)]{shepherdson59}
John~C. Shepherdson.
\newblock The reduction of two-way automata to one-way automata.
\newblock \emph{{IBM} Journal of Research and Development}, 3\penalty0
  (2):\penalty0 198--200, 1959.
\newblock \doi{10.1147/rd.32.0198}.

\bibitem[Straubing(2018)]{StraubingSIGLOG}
Howard Straubing.
\newblock First-order logic and aperiodic languages: a revisionist history.
\newblock \emph{{ACM} {SIGLOG} News}, 5\penalty0 (3):\penalty0 4--20, 2018.
\newblock \doi{10.1145/3242953.3242956}.

\end{thebibliography}
\label{sec:biblio}

\end{document}